\documentclass[a4paper,twocolumn,notitlepage,nofootinbib,superscriptaddress,floatfix]{revtex4-2}
\usepackage{amsmath}
\usepackage{amssymb}
\usepackage{amsthm}
\usepackage{graphicx}
\usepackage{xcolor}
\usepackage[english]{babel}
\usepackage{csquotes}
\usepackage{braket}
\usepackage{hyperref}
\usepackage{mathtools}
\usepackage{enumitem}
\usepackage{IEEEtrantools}
\usepackage{relsize}
\usepackage{placeins}

\usepackage[caption=false]{subfig}

\usepackage{tikz-cd} 
\usepackage{adjustbox}

\newcommand{\diff}{\ensuremath\mathrm{d}}

\newcommand*{\Tr}{\operatorname{Tr}}

\providecommand{\myvec}[1]{\ensuremath{\boldsymbol{#1}}}

\providecommand{\gg}{\ensuremath{\myvec{g}}}

\providecommand{\xx}{\ensuremath{\myvec{x}}}
\providecommand{\yy}{\ensuremath{\myvec{y}}}

\providecommand{\aalpha}{\ensuremath{\myvec{\alpha}}}
\providecommand{\bbeta}{\ensuremath{\myvec{\beta}}}
\providecommand{\ggamma}{\ensuremath{\myvec{\gamma}}}

\providecommand{\ttheta}{\ensuremath{\myvec{\theta}}}

\providecommand{\ssigma}{\ensuremath{\myvec{\sigma}}}

\providecommand{\calD}{\ensuremath{\mathcal{D}}}
\providecommand{\calE}{\ensuremath{\mathcal{E}}}

\providecommand{\calG}{\ensuremath{\mathcal{G}}}
\providecommand{\calH}{\ensuremath{\mathcal{H}}}

\providecommand{\calL}{\ensuremath{\mathcal{L}}}

\providecommand{\calS}{\ensuremath{\mathcal{S}}}
\providecommand{\calT}{\ensuremath{\mathcal{T}}}

\providecommand{\calX}{\ensuremath{\mathcal{X}}}
\providecommand{\calY}{\ensuremath{\mathcal{Y}}}

\providecommand{\bbC}{\ensuremath{\mathbb{C}}}

\providecommand{\bbI}{\ensuremath{\mathbb{I}}}

\providecommand{\bbR}{\ensuremath{\mathbb{R}}}

\providecommand{\bbZ}{\ensuremath{\mathbb{Z}}}

\providecommand{\fraks}{\ensuremath{\mathfrak{s}}}

\providecommand{\fraku}{\ensuremath{\mathfrak{u}}}

\newcommand{\ie}{\textit{i.e.}\ }
\newcommand{\eg}{\textit{e.g.}\ }

\newtheorem{theorem}{Theorem}

\newtheorem{proposition}[theorem]{Proposition}

\newcommand{\SWAP}{\ensuremath{\operatorname{SWAP}}}
\newcommand{\CNOT}{\ensuremath{\operatorname{CNOT}}}

\DeclareMathOperator{\SO}{\ensuremath{SO}}
\renewcommand{\O}{\ensuremath{\operatorname{O}}}
\DeclareMathOperator{\so}{\ensuremath{\mathfrak{so}}}
\DeclareMathOperator{\SU}{\ensuremath{SU}}
\DeclareMathOperator{\su}{\ensuremath{\mathfrak{su}}}

\DeclareMathOperator{\ad}{\ensuremath{ad}}
\DeclareMathOperator{\Ad}{\ensuremath{Ad}}

\renewcommand{\gg}{\boldsymbol{g}}

\newcommand{\fu}{Dahlem Center for Complex Quantum Systems, Freie Universit\"{a}t Berlin, 14195 Berlin, Germany}
\newcommand{\porsche}{Porsche Digital GmbH, 71636 Ludwigsburg, Germany}
\newcommand{\hzb}{Helmholtz-Zentrum Berlin f{\"u}r Materialien und Energie, 14109 Berlin, Germany}
\newcommand{\hhi}{Fraunhofer Heinrich Hertz Institute, 10587 Berlin, Germany}

\begin{document}

\title{Exploiting symmetry in variational quantum machine learning}
\date{May 12, 2022}

\author{Johannes Jakob Meyer}
\affiliation{\fu}

\author{Marian Mularski}
\affiliation{\fu}
\affiliation{\porsche}

\author{Elies Gil-Fuster}
\affiliation{\fu}
\affiliation{\hhi}

\author{Antonio Anna Mele}
\affiliation{\fu}

\author{Francesco Arzani}
\affiliation{\fu}

\author{Alissa Wilms}
\affiliation{\fu}
\affiliation{\porsche}

\author{Jens Eisert}
\affiliation{\fu}
\affiliation{\hzb}
\affiliation{\hhi}

\begin{abstract}
    Variational quantum machine learning is an extensively studied application of near-term quantum computers. 
    The success of variational quantum learning models crucially depends on finding a suitable parametrization of the model that encodes an inductive bias relevant to the learning task. 
    However, precious little is known about guiding principles for the construction of suitable parametrizations. 
    In this work, we holistically explore when and how symmetries of the learning problem can be exploited to construct quantum learning models with outcomes invariant under the symmetry of the learning task. 
    Building on tools from representation theory, we show how a standard gateset can be transformed into an equivariant gateset that respects the symmetries of the problem at hand through a process of gate symmetrization. 
    We benchmark the proposed methods on two toy problems that feature a non-trivial symmetry and observe a substantial increase in generalization performance. 
    As our tools can also be applied in a straightforward way to other variational problems with symmetric structure, we show how equivariant gatesets can be used in variational quantum eigensolvers.
\end{abstract}

\maketitle

The advent of quantum devices that come close to outperforming classical computers in certain computational tasks~\cite{arute2019quantum,wu_strong_2021} has sparked a wealth of investigations into the capabilities of \emph{noisy intermediate-scale quantum (NISQ)} devices~\cite{preskill2018quantum}. 
These endeavors are aimed at exploiting the computational power of quantum computers without quantum error correction that make use of relatively short quantum circuits. A particular emphasis is put on \emph{hybrid} quantum-classical approaches that use the quantum device to implement subroutines in otherwise largely classical algorithms~\cite{mcclean2016theory}.

Important areas of applications intended for these hybrid approaches include problems in quantum chemistry, classical
combinatorial optimization and machine learning~\cite{cerezo2021variational,BhartiDoi}.
\emph{Symmetries} have always played a very important role in the analysis of physical systems, due to their frequent admitting of conserved quantities. Combining this with Noether's theorem informs us that conserved quantities correspond to symmetries of the underlying system. In quantum chemistry, for example, solutions to ground state problems or electronic structure calculations have to respect conserved quantities such as the global particle number or the parity of fermion number and hence admit non-trivial symmetries. 

The usefulness of symmetries is not limited to the domain of physical systems, however, and their role in the context of classical machine learning must not be understated. Symmetries have been pivotal in the development of successful models for image recognition that sparked a revolution in learning with artificial neural networks~\cite{lecun1989backpropagation}. More recent breakthroughs, like the development of the transformer model~\cite{vaswani2017attention}, are also linked to a better understanding of the underlying symmetries of the learning task. Consequently, the rigorous treatment of symmetries in machine learning has recently culminated in the creation of the sub-field of \emph{geometric deep learning}~\cite{bronstein2021geometric}.

\begin{figure}
    \centering
    \includegraphics{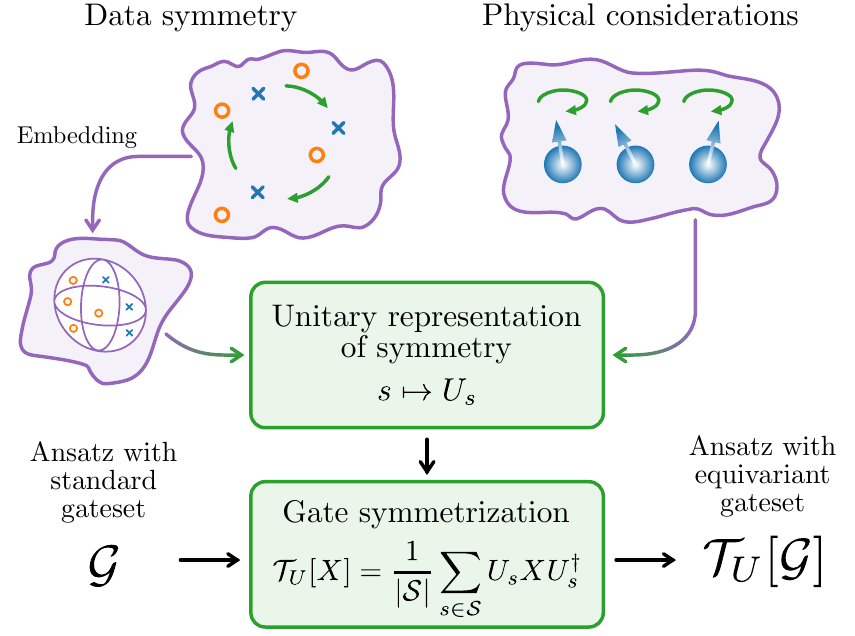}
    \caption{A unitary representation of a symmetry group $\calS$ can arise from data symmetries when the data points are suitably encoded or alternatively from physical considerations of a variational problem. We can use such a representation to replace gates with their equivariant counterparts and thus symmetrize generic ansätze used in the construction of learning models and variational algorithms.}
    \label{fig:summary}
\end{figure}

In this work, we explore how symmetries of a learning task can be exploited to create quantum learning models whose output is invariant under symmetry transformations on the level of the data. The fact that classical data has to be embedded into the Hilbert space of a quantum system makes it difficult to directly transfer ideas from the classical world to the construction of variational quantum learning models. As our main contribution, we show how to embed data in a way that enables us to build invariant learning models for which we give both a general blueprint and mathematical tools for their design.

In Sec.~\ref{sec:symmetric_structures}, we exhibit data embeddings that give rise to meaningful representations of practically relevant data symmetries on the level of the Hilbert space, both for discrete and continuous symmetries. 
The trainable parts of variational quantum learning models are usually constructed using gates from a standard gateset.
In Sec.~\ref{sec:gateSymmetrization} we show how the representation on the level of the Hilbert space can be exploited to symmetrize the generators of such a gateset to give rise to an equivariant gateset. We especially detail which pitfalls should be avoided when constructing an equivariant gateset.
In Sec.~\ref{sec:invariant_models} we show how these building blocks can be combined to construct variational quantum data re-uploading models that make invariant predictions and thus include a meaningful inductive bias.

To elaborate on how our approaches are applied and to showcase the increase in performance, we conduct a number of numerical experiments in Sec.~\ref{sec:numerics}. We consider the 
classification of games of tic-tac-toe as a paradigmatic learning task with non-trivial symmetry and show a significant increase in generalization performance that even generically holds for random constructions in the equivariant gateset. We further consider a learning task with similar symmetry properties that is inspired by a problem from the automotive industry and which can, as soon as the quantum computational resources become available, be scaled up to address the actual task. For this task, concerned with the classification of the criticality of driving scenarios for autonomous vehicles, we reach the same conclusions as for the tic-tac-toe example. We also argue that the idea of equivariant gatesets itself is not limited to applications in quantum machine learning. This is why we additionally provide numerical investigations into ground state problems of the transverse-field Ising model, the Heisenberg model and a toy model based on the longitudinal-transverse-field Ising model that exhibits the same symmetries as the tic-tac-toe problem. For these kinds of problems, equivariant ansätze usually help, but the benefit is less apparent than in the case of quantum machine learning models.

There is already a substantial amount of literature on the application of symmetries in the context of variational quantum algorithms. Explicit constructions of ansätze and gadgets for quantum chemistry problems can be found in Refs.~\cite{liu2019variational,gard2020efficient,setia2020reducing,anselmetti2021local,arrazola2021universal,vogt2021preparing,zhang2021shallow-circuit,barron2021preserving}. Another approach has been chosen by the authors of Refs.~\cite{seki2020symmetry-adapted,seki2022spatial} where symmetries were imposed through classical post-processing. An approach in the same direction as the one we outline in this work has been proposed in Ref.~\cite{herasymenko2021diagrammatic}, where it has been suggested to simply remove generators that violate the symmetry or fix the parameters of the gates. 

Here, we are reaching the same goal via a distinctly different route, as we propose to alter the set of generators used in the construction of the ansatz. Further works that treat symmetry-enhanced variational methods are found in Refs.~\cite{selvarajan2022variational,lyu2022symmetry,ParkBrokenSymm}. Another possibility outside of altering the ansatz itself to enforce symmetries is to enforce them through the optimizer via the use of additional terms in the cost function that penalize asymmetric states as in Refs.~\cite{bookatz2015error,ryabinkin2018symmetry,kuroiwa2021penalty}.
In a different direction, Ref.~\cite{mernyei2021equivariant} proposes a quantum model that closely follows the geometric deep learning blueprint~\cite{bronstein2021geometric} by having first a symmetry equivariant quantum \enquote{layer}, followed by a symmetry invariant classical \enquote{aggregator}.
In particular, they propose a model for graph-based learning tasks, so the symmetry they consider is the permutation (or relabeling) of the graph nodes.
Also specifically for graphs, Ref.~\cite{verdon2019quantum} included a recipe for building permutation invariant models inspired by \emph{graph convolutional neural networks}~\cite{bronstein2021geometric}. In Ref.~\cite{zheng2021speeding}, the authors look at the construction of group-equivariant ansatz states to learn quantum states. 
Ref.~\cite{glick2022covariant} proposes a construction of covariant quantum kernels that evaluate data that already comes with a group structure.

\section{Preliminaries}\label{sec:preliminaries}

\subsection*{Variational quantum algorithms}

\emph{Variational quantum algorithms (VQAs)}~\cite{cerezo2021variational} are a major research direction in the study of the capabilities of near-term quantum devices to solve practically relevant problems. In these hybrid quantum-classical approaches~\cite{mcclean2016theory}, a quantum device uses a \emph{parametrized quantum circuit (PQC)} as an \emph{ansatz} to prepare a parametrized quantum state vector $\ket{\psi(\ttheta)}$. After the preparation of the state, measurements are performed in order to estimate the expectation values of a set of observables $\{ O_1, O_2, \dots \}$. These expectation values are then used to compute a \emph{cost function} $C(\ttheta, \langle O_1 \rangle, \langle O_2 \rangle, \dots )$ that encodes the problem of interest. Particularly prominent applications are the \emph{variational quantum eigensolver (VQE)} where the cost function is given by the expectation value of a Hamiltonian $H$ whose ground state is to be approximated and the \emph{quantum approximate optimization algorithm (QAOA)} which combines a specific ansatz construction with VQE to approximate solutions to classical optimization problems. As variational quantum algorithms offer an extremely flexible paradigm to tackle problems, such algorithms have been proposed for a wide variety of different problems beyond that~\cite{BhartiDoi}.

\subsection*{Variational quantum learning models}
With the current success of machine learning, especially neural-network based deep learning, the interest into possible gains from applying quantum computers to learning problems has come into focus, sparking the field of \emph{quantum machine learning (QML)}. Variational approaches as outlined above can be used to construct \emph{variational quantum learning models (VQLMs)}, see e.g.\ Refs.~\cite{mitarai2018quantum,benedetti2019parameterized,schuld2020circuit-centric}. In such circuit-based models, the ansatz state does not only depend on variational parameters but also on the input data, and predictions are encoded in expectation values of observables
\begin{align}
    y(\xx) =\langle{\psi(\ttheta, \xx)} | O | {\psi(\ttheta, \xx)} \rangle.
\end{align}
Because of the no-cloning theorem, it is important for the expressivity of quantum learning models to embed the data into the quantum circuit multiple times, a process dubbed \emph{data re-uploading}~\cite{perez-salinas2020data,schuld2021effect,jerbi2022quantum}. In the following, we will therefore consider VQLMs where a fixed data-embedding unitary $U(\xx)$ is interleaved with parametrized quantum circuits as
\begin{align}\label{eqn:re-uploading-model}
    \ket{\psi(\ttheta, \xx)} =  W_L(\ttheta) U(\xx) \dots W_1(\ttheta) U(\xx) W_0(\ttheta) \ket{\psi_0}.
\end{align}

\subsection*{Symmetry groups}
Symmetries are commonly captured by \emph{groups}. A group $\calS$ is a set of objects $s \in \calS$ together with a composition rule $\circ$ -- usually called \enquote{multiplication} -- so that $s_1 \circ s_2 \in \calS$ for all $s_1, s_2 \in \calS$. The multiplication is also required to be associative, \ie $(s_1 \circ s_2) \circ s_3 = s_1 \circ (s_2 \circ s_3)$. Furthermore, there is an identity element $e$ so that $e \circ s = s \circ e = s$ for all $s \in \calS$ and that for all $s \in \calS$ there exists a unique inverse element $s^{-1}$ so that $s \circ s^{-1} = s^{-1} \circ s = e$. A simple example of a group is $\bbZ_2 = \{ 0, 1\}$ where the composition is given by addition modulo 2.
If a restriction of $\calS$ to some subset is itself a group under the composition rule of $\calS$, we call this subset a \emph{subgroup}.

The concept of groups also extends to continuous sets in the form of \emph{Lie groups}. A real Lie group is a group $\calS$ which is also a smooth manifold and the composition rule and inversion are smooth maps. One example of a (non-compact) Lie group is the set of all invertible linear operators $\operatorname{GL}_n(\bbC)$. Lie groups that are a subgroup of $\operatorname{GL}_n(\bbC)$ and can hence be expressed in terms of matrices are called \emph{matrix Lie groups} and we will focus on these in the following. An example of a matrix Lie group are the unitary matrices acting on $\bbC^{n}$, denoted as $\operatorname{U}(n)$.

Lie groups are intimately related to \emph{Lie algebras} which can be seen as the generators of elements of the Lie group. For a matrix Lie group $\calS$, we can define the associated algebra as
\begin{align}
    \fraks = \operatorname{Lie}(\calS) = \{ G \colon \exp(tG) \in \calS \text{ for all } t\in \bbR \}.
\end{align}
This is motivated by the fact that for two group elements $s_1, s_2 \in \calS$ we can understand the group multiplication on the level of the Lie algebra via the Baker-Campbell-Hausdorff formula
\begin{align}
    s_1 s_2 &= e^{G_1}e^{G_2} = e^{G_1 + G_2 + \frac{1}{2}[G_1, G_2]+ \dots}.
\end{align}
For $s_1$ and $s_2$ close to the identity element, we have that the norm of $G_1$ and $G_2$ are small in which case their commutator gives the first correction. It is actually part of the formal definition of a Lie algebra and is called the Lie bracket. 

A \emph{representation} of a group is a map $U\colon \calS \to \operatorname{Aut}(V), s \mapsto U_s$ that is compatible with the composition rule: $U_{s_1 \circ s_2} = U_{s_1} U_{s_2}$. If the space $V$ is equipped with a scalar product and the linear maps $U_s$ are unitary, we call this a unitary representation.

\section{Symmetries induced by data embeddings}\label{sec:symmetric_structures}

Symmetries play an important role all across physics and problems that are of interest in near-term applications are no different in this respect. A symmetry group $\calS$ acts on an $n$-qubit Hilbert space $\calH$ through a unitary representation $U_s$ for $s \in \calS$.

In ground state problems -- as they are encountered in the variational quantum eigensolver~\cite{peruzzo2014variational} -- symmetries are usually derived from physical considerations. An example of such a symmetry, which is encountered for example in the Heisenberg-model, is that the energy of a state does not change if all spins are flipped, i.e.\ when the operator $X^{\otimes n}$ is applied. This can be understood as a representation of the symmetry group $\bbZ_2$. Further symmetries include joint unitary transformations of all spins ($\operatorname{SU}(2)$) or translation symmetry ($\bbZ_n$). Such symmetries have already been extensively studied in the context of VQE, with a particular focus towards building symmetry-preserving ansätze for quantum chemistry applications~\cite{liu2019variational,gard2020efficient,setia2020reducing,anselmetti2021local,arrazola2021universal,vogt2021preparing,zhang2021shallow-circuit,barron2021preserving}. 

\begin{figure*}
    \centering
    \includegraphics{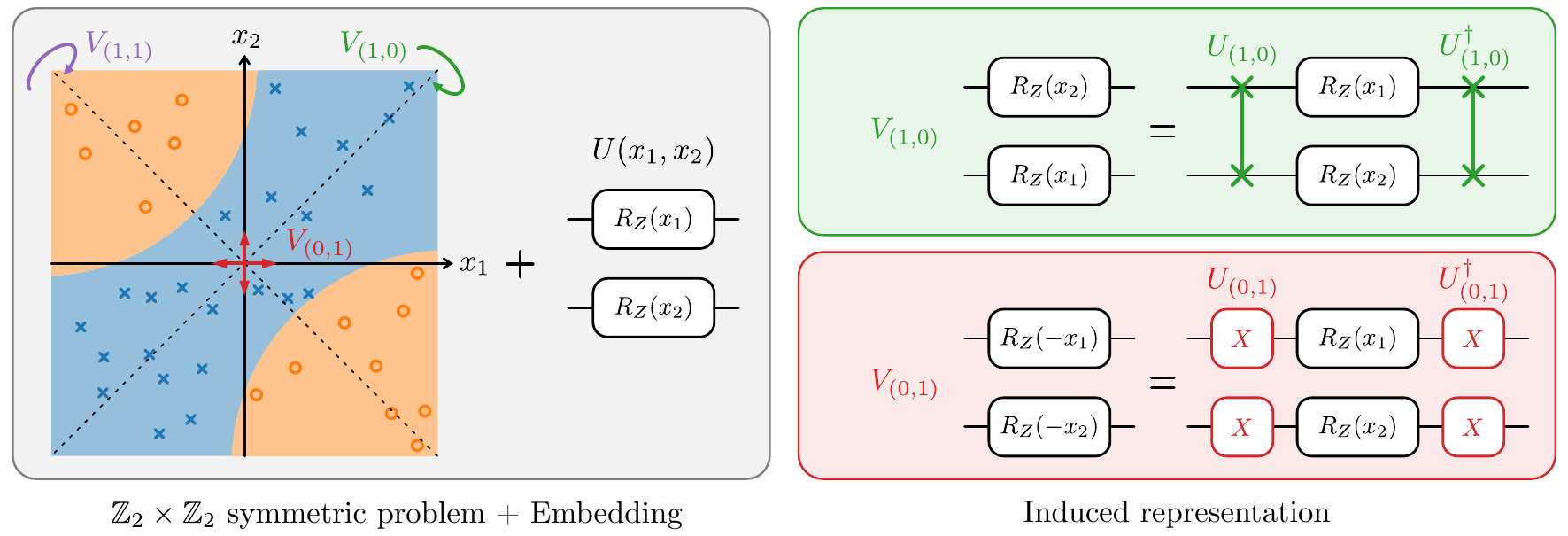}
    \caption{Simple example of a classification problem with two features $(x_1, x_2)$ that has a ${\bbZ_2 \times \bbZ_2}$ symmetry generated by a flip about the axis $x_1 = x_2$ (green) and inversion about the center (red). Combining both operations yields a flip about the axis $x_1 = -x_2$ (purple). Together with the simple qubit-wise embedding of the coordinates, these induce a representation of ${\bbZ_2 \times \bbZ_2}$ generated by a swap operation of the two qubits and a simultaneous application of Pauli $X$ on the two qubits.}
    \label{fig:vqlm_example_Z2}
\end{figure*}

But ground state problems are not the only problems of interest in near-term applications that can have a symmetric structure. 
Indeed, many learning problems encountered in the context of quantum machine learning also have some symmetry: A prototypical example is the labeling of images. Even if we move all the pixels of a cat photo to the right or rotate it, the photo still depicts a cat. More formally, consider a task where a prediction $y \in \calY$ should be associated to data points $\xx \in \calX$. We say that the prediction is invariant under a symmetry group $\calS$ with representation $V_s\colon \calS \to \operatorname{Aut}(\calX)$ if 
\begin{align}\label{eqn:invariance_of_prediction}
    y(V_s[\xx]) = y(\xx)  \text{ for all } \xx \in \calX \text{ and } s \in \calS.
\end{align}
Notions like this have recently sparked the subfield of \emph{geometric deep learning}~\cite{bronstein2021geometric} which abstractly reasons about these symmetries and their implementation in learning models, particularly in artificial neural networks. The geometric reasoning also elucidates the early success of convolutional neural networks~\cite{lecun1989backpropagation} which use building blocks that respect the translation symmetry present in typical image classification tasks and the success of transformer models~\cite{vaswani2017attention}. It is also intuitive to expect that building symmetry-invariant models comes with advantages: By definition, a symmetry-invariant model does not see any difference between data points that are related by a symmetry transformation which effectively reduces the number of possible predictions the learning model can make. Such a model has lower complexity and is thus expected to be easier to train. Furthermore, these models only produce predictions that respect the symmetry already present in the learning problem and are thus expected to generalize better to unseen data.

In the following, we will show how to construct quantum circuits that embed classical data in a way that induces a meaningful unitary representation of the symmetry on the level of the Hilbert space. These embeddings are then used as building blocks in the construction of symmetry-invariant data re-uploading models.

\subsection*{Introductory example}
We first consider an introductory example that is given in Figure~\ref{fig:vqlm_example_Z2}. 
Assume we have data points $\xx = (x_1, x_2) \in \bbR^2$ and the prediction -- for example a classification -- has the following symmetries
\begin{align}
    y(x_1, x_2) = y(x_2, x_1) = y(-x_1, -x_2)
\end{align}
which correspond to an exchange of the coordinates and a simultaneous sign flip. The associated symmetry group, famously known as Klein's four-group, is given by
\begin{align}
    \bbZ_2 \times \bbZ_2 = \{(0, 0), (1,0), (0,1), (1, 1)\}
\end{align}
with the group operation being entry-wise addition modulo 2. In this abstract definition of the symmetry group, the first component could be understood as the boolean answer to the question \enquote{do we exchange?}, and the second component as the boolean answer to the question \enquote{do we flip?}. 
The representation on the level of the data is given by
\begin{IEEEeqnarray}{ll}
    V_{(0,0)} = \begin{bmatrix}1 & 0\\0 & 1\end{bmatrix}, &
    \quad V_{(1,0)} = \begin{bmatrix}0 & 1\\1 & 0\end{bmatrix} ,\\
    V_{(0,1)} = \begin{bmatrix}-1 & 0\\0 & -1\end{bmatrix}, &
    \quad V_{(1,1)} = \begin{bmatrix}0 & -1\\-1 & 0\end{bmatrix} .
\end{IEEEeqnarray}
If we use an embedding unitary
\begin{align}
    U(x_1, x_2) = R_Z(x_1) \otimes R_Z(x_2),
\end{align}
then the symmetry operations can be represented on the level of the Hilbert space as
\begin{align}
    U(x_2, x_1) &= R_Z(x_2) \otimes R_Z(x_1)  \\
    &= \SWAP U(x_1, x_2) \SWAP,  \\
    U(-x_1, -x_2) &= R_Z(-x_1) \otimes R_Z(-x_2)  \\
    &= (X\otimes X) U(x_1, x_2) (X \otimes X),
\end{align}
where we have used the fact that $XZX = -Z$. Note that we could have equivalently used $(Y \otimes Y)$ in the above example, so the operations are not necessarily unique.

We can use the intuition we gathered from the example to generalize the concept of an \emph{induced representation}: We call a data encoding unitary $U(\xx)$ \emph{equivariant} with respect to the data symmetry $V_s$\footnote{We choose \emph{equivariant} over the synonymous \emph{covariant} as this is the terminology used in geometric deep learning.} if
\begin{align}\label{eqn:embedding_symmetry_condition}
    U(V_s[\xx]) = U_s U(\xx) U_s^{\dagger}
\end{align}
for a unitary representation $U_s$ of $\calS$. The representation induced by our above example is then
\begin{IEEEeqnarray}{ll}
    U_{(0,0)} = \bbI \otimes \bbI ,&
    \quad U_{(1,0)} = \SWAP ,\\
    U_{(0,1)} = X \otimes X, &
    \quad U_{(1,1)} = \SWAP (X \otimes X).
\end{IEEEeqnarray}
In the following, we will show how we can construct data encodings which are equivariant with respect to the most important symmetry classes encountered in real-world learning tasks. 

\subsection*{Permutation symmetries}

Most data encountered in contemporary learning tasks, like images or time series data, has a discrete structure. Symmetry transformations of the data, like translating or rotating an image or translating a time series, can thus be captured by a permutation of the different data features. This means that in many relevant scenarios with symmetry, the symmetry group will be a subgroup of the \emph{permutation group} (or \emph{symmetric group}) $S_d$ of the $d$ different data features $\xx=(x_1, \ldots, x_d)$.

The elements of the symmetric groups are all possible reshufflings $\sigma$ of the data features
\begin{align}
    \sigma(x_1, x_2, \ldots, x_d) &= (x_{\sigma(1)}, x_{\sigma(2)}, \ldots, x_{\sigma(d)}).
\end{align}
Every permutation can be constructed by concatenating \emph{transpositions}, which are exchanges of two data features at a time. We will write a transposition as $\tau$ such that for all $\sigma \in S_d$ we have that $\sigma = \tau_m \dots \tau_2 \tau_1$ for some transpositions $\{ \tau_i \}$.
Note that a product of transpositions is usually considered to be implemented from right to left.

In the previous section, we already saw how to construct a data encoding which is equivariant with respect to the permutation of two elements, and the construction directly generalizes to higher dimensions, needing $d$ qubits for $d$-dimensional data. We complement this with a related construction that allows using approximately $\log d$ qubits but still retains equivariance. 
The natural construction is to embed every data feature via a Pauli rotation on a separate qubit, the choice of the generator is not relevant, so we settle for $Z$, to get
\begin{align}
    U(x_1,\ldots, x_d) &= R_Z(x_1) \otimes \cdots\otimes R_Z(x_d).
\end{align}
This encoding is obviously equivariant with respect to the permutations of the qubits and induces a representation of $S_d$ which is given by
\begin{align}
   U_\sigma = U_{\tau_m \dots \tau_2 \tau_1} = \SWAP_{\tau_m}  \dots \SWAP_{\tau_2} \SWAP_{\tau_1}
\end{align}
where $\SWAP_{\tau}$ swaps the two qubits corresponding to the features swapped by $\tau$. This embedding actually allows us to go even further than the permutation group and also equivariantly embed sign-flips of the $i$-th coordinate by conjugating with Pauli $X$ on the $i$-th qubit, thus generalizing to $S_d \rtimes \bbZ_2^d$.\footnote{
    This new group is the \emph{semidirect product} $S_d \rtimes\bbZ_2^d$ which captures the intuition that every element $g\in S_d\rtimes\bbZ_2^d$ can be written as the product $g=f\circ\sigma$ of one permutation $\sigma\in S_d$ and a number of sign-flips $f\in\bbZ_2^d$.
}

With this strategy, we need as many qubits as data features.
A natural question is whether it is possible to do better, namely if we could find a different encoding gate that embedded $d$ features in less than $d$ qubits and still be equivariant with respect to the permutation group.
If we want to add another Pauli-word as a generator to the above set of local $Z$ operators, we are limited by the fact that this generator has to be invariant under all $\SWAP$ operations between qubits because these should by definition only exchange the associated coordinates and not the newly added one. This means the only operator we can add is $Z^{\otimes d}$, yielding one additional data feature we could in principle encode with this strategy. In this case, the representation for exchanging the $i$-th coordinate with the $d+1$-st coordinate is a multi-$\CNOT$ gate where the $i$-th qubit controls a simultaneous application of $X^{\otimes d-1}$ on the rest of the qubits.

Note that we could in principle employ a different set of mutually commuting Pauli words as any set of independent, mutually commuting Pauli words are interchangeable via Clifford-gates~\cite{aaronson2004improved}. One example would be the set of generators $\{ Z \otimes \bbI, \bbI \otimes Z, Z \otimes Z\}$ which could equivalently be replaced by $\{ X \otimes X, Y \otimes Y, Z \otimes Z \}$, which would constitute a more \enquote{balanced} encoding strategy as the operators are of similar weight. For this set of operators, the permutations are generated by the transpositions
\begin{align}
    \tau_{12} &\leftrightarrow H_{XY} \otimes H_{XY}, \\
    \tau_{23} &\leftrightarrow H_{YZ} \otimes H_{YZ} ,\\
    \tau_{31} &\leftrightarrow H_{ZX} \otimes H_{ZX}, 
\end{align}
where $H_{AB}$ is a generalization of the Hadamard operator $H = H_{ZX}$ given by $H_{AB} = \exp(-i \pi  (A + B)/2\sqrt{2})$. 

The encoding based on local $Z$-rotations is very intimately related to the IQP encoding proposed in Ref.~\cite{havlicek2019supervised}. We can use their inspiration to also include higher monomials of the data features like for example $\xi = x_1 x_4$. If we embed such a monomial with a generator $Z_1 Z_4$, where $Z_i$ is a local $Z$ acting on the $i$-th qubit, then we see that a permutation of the features can be realized by a permutation of the qubits. An equivariant embedding of the data is then achieved by simultaneously embedding all possible permutations of monomial like $x_1 x_2$, $x_1 x_3$, $x_2 x_3$ and so on using the operators $Z_1 Z_2$, $Z_1 Z_3$, $Z_2 Z_3$ etc. The number of terms needed for this construction grows rather quickly because all possible permutations have to be considered. On the other, this means that we can expect to have to embed fewer terms if we only wish to realize a subgroup of the permutation group. If we consider only translations that map $x_i x_j$ to $x_{i+1} x_{j+1}$, we only need to embed all possible translations of the monomial which is linear in the number of features.

While the above strategy for embedding data is straightforward to implement, we can in principle embed exponentially many data features in an equivariant way. One strategy to do so is to assign every state of the computational basis as the collection of all $\ket{k}$ identified by the binary representation of the integer $k$ a generator
\begin{align}
    G_k = |k \rangle \!\langle k |.
\end{align}
This generates unitaries $R_k(x_k) = \exp(-i x_k |k \rangle \!\langle k |)$ which only put a phase on the amplitudes of the $k$-th computational basis state. The transposition of features $k$ and $l$ is then easily seen to be achieved by
\begin{align}
    \tau_{kl} \leftrightarrow |k \rangle \!\langle l | + |l \rangle \!\langle k | + \sum_{j \neq k,l} |j \rangle\!\langle j |,
\end{align}
which induces a representation of the permutation group for a number of data features up to $2^n$. However, contrary to the simple $Z$ rotations in the example above, these rotations are very difficult to implement, as they mostly correspond to operations controlled on nearly all qubits. Furthermore, as we will see in Sec.~\ref{sec:gateSymmetrization}, if all of the available basis states are used, only trivial operations will respect the symmetry of the system.

The above example motivates that the construction of equivariant embeddings should be possible with relatively few qubits, but at the cost of very complicated gates that need to be implemented. This reduces the practicality of such an approach, especially in the NISQ era. 
The other extreme is the 1-local embedding proposed at the beginning of this section, which essentially only allows us to embed as many data features as there are qubits but which can be implemented very easily. 
It is therefore an interesting direction for future research what kind of embeddings can be defined that interpolate between these two regimes and allow to embed more data features than the 1-local embedding while still being sufficiently implementable.

\subsection*{Continuous symmetries}
    
While many cases of practical interest display discrete symmetries, we need not constrain ourselves to this case. Especially when looking at possible applications in the natural sciences we can also encounter continuous symmetries, if for example rotating spatial data points around a certain axis does not alter predictions. The framework for dealing with continuous symmetries is given by Lie groups which are continuous groups with a differentiable structure.  

Representation theory tells us that any compact Lie group can be expressed as a subgroup of a suitably large unitary group~\cite{folland2016course}, and as any unitary group can be embedded in a larger special unitary group every Lie group can be represented on a suitably large quantum system.  
However, it is not a priori clear which Lie groups admit embeddings that induce a faithful unitary representation. 
We will show how we can construct an embedding that is equivariant with respect to the special orthogonal group in $3$ dimensions, $\SO(3)$, and show how we can extend this to the full orthogonal group $\O(3)$. We then outline how the general task can be understood mathematically.

\subsubsection*{$\SO(3)$ and $\O(3)$}
    
The group $\SO(3)$ is made of all possible rotations of a sphere centered at the origin.
For elements of the special orthogonal group, we will use lower case letters to distinguish them from their quantum counterparts.
Any element $r \in \SO(3)$ can be constructed by successively applying one of the three canonical rotations about the axes of the coordinate system which we denote as $r_x(\alpha), r_y(\alpha)$ and $r_z(\alpha)$.
In the following, we will use the parametrization in terms of the Euler angles given by
\begin{align}
    r(\psi,\theta,\phi) &= r_z(\psi)r_x(\theta)r_z(\phi).
\end{align}
We now define an equivariant embedding of a data point $\xx\in\bbR^3$ as
\begin{align}
    U(\xx) &= e^{-\frac{i}{2}(x_1X+x_2Y+x_3Z)} 
    = e^{-\frac{i}{2} \langle \xx, \ssigma \rangle},
\end{align}
where we defined the vector of Pauli operators ${\ssigma=(X,Y,Z)}$.
We can now exploit that this is essentially a mapping from data points to the Bloch sphere, of which we know that conjugation with a Pauli rotation generates the associated rotation of the Bloch sphere to obtain the induced representation. We only need that
\begin{align}
    \langle r_x(\alpha) \xx, \ssigma \rangle &= \langle \xx, r_x(-\alpha) \ssigma \rangle\\
    &= R_X(-\alpha) \langle \xx, \ssigma \rangle R_X(\alpha),\\
    \langle r_z(\alpha) \xx, \ssigma \rangle &= R_Z(-\alpha) \langle \xx, \ssigma \rangle R_Z(\alpha),
\end{align}
to deduce that
\begin{align}
    ~U(r(\psi, \theta, \phi) \xx) 
     =&~ e^{-\frac{i}{2}\langle \xx, r(-\phi, -\theta, -\psi) \ssigma \rangle} \\
     =&~ e^{-\frac{i}{2}\langle \xx, r_z(-\phi)r_x(-\theta)r_z(-\psi) \ssigma \rangle} \\
     \begin{split}
     =&~ R_Z(-\psi) R_X(-\theta) R_Z(-\phi) e^{-\frac{i}{2}\langle \xx, \ssigma \rangle}\\
     &\qquad  R_Z(\phi) R_X(\theta)R_Z(\psi).
     \end{split}
\end{align}
At this point, we recall that there is a parametrization for arbitrary unitaries $U\in\SU(2)$, also in terms of three angles $U(\psi,\theta,\phi) = R_Z(\psi)R_X(\theta)R_Z(\phi)$. Using this, we arrive at the desired induced representation
\begin{align}
    U(r(\psi, \theta, \phi) \xx) &= U(-\psi,-\theta,-\phi) U(\xx) U^{\dagger}(-\psi,-\theta,-\phi).
\end{align}
The full group of all orthogonal (\ie length-preserving) transformations on $\bbR^d$, $\O(3)$, is made of not only rotations but also reflections about any plane that passes through the origin.
Luckily, reflections about different planes can be related by the rotation that maps the planes to one another.
This means one fixed reflection, together with the three canonical rotations is enough to span the whole group $\O(3)$.
We pick the plane perpendicular to the vector $(1,1,1)$, which realizes an inflection about the origin $\xx \mapsto -\xx$. As we have seen that we have already exhausted $\SU(2)$ to represent $\SO(3)$ this can not be possible using an embedding with only a single qubit. However, we can actually straightforwardly realize this inflection if we add an additional qubit and use the embedding
\begin{align}
    U(\xx) = e^{-\frac{i}{2}( x_1 X + x_2 Y + x_3 Z) \otimes X}.
\end{align}
The rotations of $\SO(3)$ are embedded in the same way as before on the first qubit, but now the inflection can also be realized as
\begin{align}
    U(-\xx) &= e^{-\frac{i}{2}( -x_1 X - x_2 Y  - x_3 Z) \otimes X} \\
    &= e^{-\frac{i}{2}( x_1 X  +x_2 Y  +x_3 Z) \otimes -X} \\
    &= (\bbI \otimes Z) e^{-\frac{i}{2}( x_1 X  +x_2 Y  +x_3 Z) \otimes X} (\bbI \otimes Z) \\
    &= (\bbI \otimes Z) U(\xx) (\bbI \otimes Z).
\end{align}
In this way, we can generate the entire group $\O(3)$.

\subsubsection*{The general case}

We have found an encoding that is equivariant with respect to $\SO(3)$ transformations, but the immediate question arising from this construction is if it can be generalized to arbitrary Lie groups acting on the data. To this end, we will mathematically formalize the process that leads to an equivariant embedding.   
Note that it is required that the symmetry $\calS$ must have non-trivial finite-dimensional unitary representations which rules out interesting groups like $\SO(3,1)$ as they are not compact groups.

We assume that the embedding of the data is of the form
\begin{align}
    U(\xx) = e^{-i \calE(\xx)},
\end{align}
which means that we effectively embed the data into the Lie algebra of the quantum system $\fraks \fraku(2^n)$. We will further assume that the map from the data to the Lie algebra $\calE\colon \bbR^d \to \su(2^n)$ is linear. We model the symmetry transformation on the level of the data as a representation of the symmetry Lie group $\calS$ given by $V_s\colon \bbR^d \to \bbR^d$. Our aim is to find a map $W_s\colon \su(2^n) \to \su(2^n)$, such that the following diagram commutes:
\[ \adjustbox{scale=1.4,center}{
\begin{tikzcd}
\bbR^d \arrow{r}{\calE} \arrow[swap]{d}{V_s} & \mathfrak{su}(2^n) \arrow{d}{W_s} \\%
\bbR^d \arrow{r}{\calE}& \mathfrak{su}(2^n)
\end{tikzcd}
}\]
Here, $W_s$ is a transformation that implements the symmetry transformation on the level of the Lie algebra in an equivariant way. This means it needs to be a conjugation by a unitary
\begin{align}
    W_s[X] = U_s X U_s^{\dagger}
\end{align}
for some $U_s \in \SU(2^n)$ that represents the symmetry group $\calS$ acting on the data. The above is nothing but the adjoint action of $U_s$, which if $U_s = e^{H}$, is generated by the adjoint action on the level of the Lie algebra
\begin{align}
    U_s X U_s^{\dagger} = \Ad_{U_s}[X] = e^{\ad_H}[X],
\end{align}
where $\ad_H[X] = [H, X]$. Looking back to the data space, we can use that $V_s$ is a matrix representation of a Lie group and that we can write $V_s = e^{G}$ for some element of the Lie algebra of $\calS$, $G \in \fraks$.
Making the diagram commute then enforces
\begin{align}\label{eqn:induced_adjoint}
\calE e^G \overset{!}{=} e^{\ad_H} \calE,
\end{align}
which can be quickly verified to be the case if
\begin{align}
    \calE G = \ad_H \calE.
\end{align}
In other words, this means that, to construct an equivariant embedding for a given Lie group representation acting on the data, we need to find a subspace of $\su(2^n)$ -- which is then the image of $\calE$ -- such that the restriction of the adjoint representation of $\su(2^n)$ to this subspace is \emph{identical} to the representation of the symmetry Lie algebra $\fraks$ acting on the data, $G$. 

This elucidates \emph{why} we were able to construct an embedding equivariant with respect to $\SO(3)$. It is because the fundamental representation of the Lie algebra $\so(3)$ which acts on the data is equal to the adjoint representation of $\su(2)$. Additional to this, we can also see that there always exists a trivial embedding $\calE_0[X] = 0$ which induces a trivial representation of the Lie algebra. From the above reasoning, we also see that it is not enough to find a sub Lie algebra isomorphic to the one acting on the data, we really have to make sure that there is an appropriate subrepresentation of the adjoint representation of $\su(2^n)$ identical to the representation of $\fraks$ on the data space.
As we think that the realization of embeddings equivariant with respect to symmetries more exotic than $\O(3)$ is a topic of limited interest we leave the classification of equivariant embeddings that can be realized in this setup as a topic of further investigation.

\section{Gate symmetrization \label{sec:gateSymmetrization}}

As we have seen above, symmetries arise naturally in variational quantum learning models if an equivariant embedding is used. In these applications, it is paramount to construct the trainable parts of the model in a way that encodes an inductive bias suitable to the problem at hand. The problem we face is that the knowledge of the relation between parametrized quantum circuits and the associated inductive bias is not really understood, leaving us with precariously little to inform our construction of learning models. Symmetries provide a first avenue to the construction of better quantum learning models as they allow us to include knowledge about the underlying data into the model in a meaningful way. They furthermore allow us to reduce the complexity of the ansatz as measured by the number of free parameters and thus also save resources. The same holds true for other variational applications, \eg in the search for ground states. There, the goal is not a better generalization capability but a better expressivity in the relevant parts of the Hilbert space.

In this section, we explain how we can use elementary group theory to construct an \emph{equivariant gateset} from a standard gateset used in ansatz constructions, where, as we will make formal below, we define equivariant as commuting with the symmetry representation.
This allows us to take an existing ansatz and to make it equivariant by replacing every gate by its equivariant counterpart. We will also explore why this approach has its advantages but is no panacea as there are several pitfalls to avoid. This sometimes makes it advisable to not build equivariance with respect to the full symmetry group into the model but to only consider subgroups thereof.

\subsection*{Equivariant gatesets}

We will focus on gates generated by a fixed generator $G$ as
\begin{align}\label{eqn:parametrized_gate_def}
    R_G(\theta) = \exp( -i \theta G),
\end{align}
as they are usually encountered in ansatz constructions for variational approaches. When constructing an ansatz, gates are picked with generators from a fixed gateset $G \in \calG$. 

We call a gate \emph{equivariant} with respect to the symmetry embodied by a unitary representation $U_s$ of $s\in \calS$ if the order of applying the symmetry operation and the gate itself does not matter, so that
\begin{align}
    [R_G(\theta), U_s] = 0 \text{ for all } \theta \in \bbR, s \in \calS.
\end{align}
This can only be the case if the generator itself commutes with the representation which is captured by the following proposition.

\begin{proposition}[Commuting generators]
For a given gate $R_G$ we have that
\begin{align}
    [R_G(\theta), U_s] = 0 \text{ for all } \theta \in \bbR, s \in \calS
\end{align}
if and only if $[G, U_s] = 0$ for all $s \in \calS$.
\end{proposition}
\begin{proof}
First, we show that the condition is necessary. For this consider the expansion of $R_G[\theta]$ to first order
\begin{align}
    0 &= [R_G(\theta), U_s] \\
    &= [\bbI - i  \theta G +  O(\theta^2), U_s] \\
    &= i  \theta [G, U_s] + O( \theta^2) \\
    \Leftrightarrow \ 0 &= [G, U_s]
\end{align}
as this relation needs to be valid also for infinitesimally small $\theta$. The condition is obviously sufficient as $[G, U_s] = 0$ implies that all powers of $G$ and hence the full exponential commutes with $U_s$.
\end{proof}

Luckily for us, there is a straightforward way to ensure that a generator does commute with a given representation -- we can make use of the \emph{twirling formula}.

\begin{proposition}[Twirling formula~\cite{helsennodatequantum}]
Let $U_s$ be a unitary representation of $\calS$. Then, 
\begin{align}
    \calT_U[X] = \frac{1}{|\calS|} \sum_{s \in \calS} U_s X U_s^{\dagger}
\end{align}
defines a projector onto the set of operators commuting with all elements of the representation, i.e., $[\calT_U[X], U_s] = 0$ for all $X$ and $s \in \calS$.
\end{proposition}
The same holds true for Lie groups if we replace the uniform average with an integration over the Haar measure $\mu$,
\begin{align}
    \calT_U[X] = \int \diff \mu(s) \, U_s X U_s^{\dagger}.
\end{align}
We use this approach to associate to any gateset $\calG$ an \emph{equivariant gateset}
\begin{align}\label{eqn:def_equivariant_gateset}
    \calT_U[\calG] = \{ \calT_U[G] \, | \, G \in \calG \}.
\end{align}

Note that this approach also directly extends to gates which are not parametrized, as these gates can either be directly symmetrized via the twirling formula or they can be written as a parametrized gate with fixed evolution angle as in Eq.~\eqref{eqn:parametrized_gate_def}.

\subsection*{Ansatz symmetrization}

We can now use the aforementioned gate symmetrization technique to convert a complete ansatz (or trainable block) to an equivariant ansatz (or equivariant trainable block) by replacing the ansatz' gateset with its equivariant counterpart. The computation of the equivariant gateset can for practical purposes usually be done efficiently beforehand.

To make this more palpable, let us return to our example with an exchange symmetry on two qubits. Assume an ansatz is made up of local rotation gates generated by the Pauli operators, $\{X ,Y,Z\}$, and entangling gates generated by a $ZZ$ interaction, all of which we consider as trainable. In this case, the gateset is
\begin{align}
    \calG = \{X_1, Y_1, Z_1, X_2, Y_2, Z_2, Z_1 Z_2\},
\end{align}
where the index identifies the qubit the Pauli is acting on. 
The $R_{ZZ}$ gate already commutes with the swap operation, which means we can focus on the local operations only. 

In our example, the symmetry group is $\bbZ_2 \times \bbZ_2$, where the symmetries were generated by $\SWAP$ (exchange $\leftrightarrow$) and $X \otimes X = X_1 X_2$ (sign flip $\pm$). We will first consider the symmetrization over the two subgroups only to synergize afterwards.

The generator $Z_1 Z_2$ already commutes with the swap operations and thus stays invariant under symmetrization. If we apply the symmetrization with respect to the exchange symmetry to $X_1 = X \otimes \bbI$, we obtain
\begin{align}
    \calT_{U_{\leftrightarrow}}[X_1] &= \calT_{U_{\leftrightarrow}}[X \otimes \bbI] \\
    &=\frac{1}{2}\left[X \otimes \bbI + \SWAP(X \otimes \bbI)\SWAP \right]\\
    &=\frac{1}{2}\left[X \otimes \bbI + \bbI \otimes X \right]\\
    &=\frac{1}{2}\left[X_1 + X_2 \right] \\
    &= \calT_{U_{\leftrightarrow}}[X_2].
\end{align}
The symmetrization proceeds analogously for the other operators. We see that $X_1$ and $X_2$ map to the same operator upon symmetrization -- hence, the symmetrized gateset has a lower cardinality. This means that, as expected, the subspace of symmetry-preserving unitaries is smaller than the full space of unitaries and that upon symmetrization, we reduce the number of parameters, and hence the complexity, of the ansatz. The gateset equivariant with respect to the exchange symmetry is then given by symmetric entangling gates and simultaneous Pauli rotations on the two qubits that share the same angle,
\begin{align}
    \calG_{\leftrightarrow} = \left\{\frac{X_1 + X_2}{2}, \frac{Y_1 + Y_2}{2}, \frac{Z_1 + Z_2 }{2}, Z_1 Z_2 \right\}.
\end{align}
Now, let us look at the sign flip symmetry. Again, $Z_1 Z_2$ already commutes with $X_1 X_2$ and stays invariant under symmetrization. The situation is again different for the local gates. It is quite straightforward that $X_1$ and $X_2$ commute with $X_1 X_2$, but this is not true for the other gates as
\begin{align}
    \calT_{U_{\pm}}[Y_1] &= \calT_{U_{\pm}}[Y \otimes \bbI]  \\
    &= \frac{1}{2}[Y \otimes \bbI + (X\otimes X)(Y \otimes \bbI)(X \otimes X)] \\
    &= \frac{1}{2}[Y \otimes \bbI - Y \otimes \bbI] \\
    &= 0,
\end{align}
where we have used the fact that $XYX = -Y$. The calculation goes analogously for $Y_2$ and also for $Z_1$ and $Z_2$ because $XZX = -Z$. This means that the sign flip equivariant gateset looks rather different than the one for the exchange symmetry, as in this case we are only allowed local Pauli $X$ rotations
\begin{align}
    \calG_{\pm} = \left\{ X_1, X_2, Z_1 Z_2 \right\}.
\end{align}

Due to the commuting nature of the group, we can obtain the fully equivariant gateset by applying either symmetrization procedure to the other gateset which yields
\begin{align}
    \calG_{\leftrightarrow , \pm} = \left\{\frac{X_1 + X_2}{2}, Z_1 Z_2 \right\}.
\end{align}
We see that taking into account the full symmetry greatly reduces the amount of available operations, but which comes at a cost of reduced expressivity as we will also detail further in the next section.

\subsection*{Pitfalls to avoid}

Making an ansatz equivariant through symmetrization can bring considerable advantages but is by no means a panacea. There are important considerations and a series of pitfalls that one should be aware of. 

First and foremost, there is always a trade-off between the gain in equivariance and therefore specialization to the relevant subspace and the reduction in expressivity of the ansatz. This is the point where quantum machine learning and ground state problems arguably differ the most: In machine learning applications, maximal expressivity is almost always a bad thing, as it leads to overfitting and therefore bad generalization performance. Depending on the amount of available data and the specifics of the learning model, the regime of overfitting can also be reached relatively quickly. In this case, equivariant models offer a clear advantage as they not only reduce the expressivity but they do so in a way that ensures that generalization improves. 
For the preparation of ground states with a given symmetry, the picture looks somewhat different. In this case, there exists no phenomenon of overfitting \emph{per se}, as a lower energy is preferred even if the prepared state does not have the same symmetry as the ground state. 

We can conclude that the trade-off between expressivity and equivariance should always be kept in mind. To fine-tune this trade-off, it can also be advised to only realize a subset of all existing symmetries of the problem. In our example above, one could for example opt to only respect the exchange symmetry of the problem. Another way to fine-tune expressivity of the model is by including a limited number of explicitly symmetry-breaking gates, as was for example argued in Ref.~\cite{ParkBrokenSymm}. This can be especially advisable if we do not know the symmetry sector the ground state is in, so that we have to have the ability to change the sector which is impossible with purely equivariant circuits. It can furthermore happen that an equivariant circuit has an unfavorable loss landscape which can sometimes be alleviated by adding some symmetry-breaking operations.

During the symmetrization process itself, there are further pitfalls to avoid:
First, we have already seen that the symmetrization procedure can trivialize certain generators. One example is the representation of $\bbZ_2$ given by $U_0 = \bbI$ and $U_1 = X$. We can express every single-qubit unitary as a decomposition $V = R_Z(\theta_1) R_Y(\theta_2) R_Z(\theta_3)$ which corresponds to using the gateset $\calG = \{ Y, Z \}$. However, neither of these generators is equivariant with respect to this representation, i.e., $\calT_U[Y] = \calT_U[Z] = 0$, leaving us with \emph{no gates}, even though we started from a universal parametrization. This is in contrast to the fact that we have a set of unitaries compatible with the symmetries which is generated as $R_X(\theta)$. If we had chosen, for example $V =R_X(\theta_1) R_Y(\theta_2) R_Z(\theta_3)$, the trivialization would not have occurred.

Second, depending on the gates that were used to construct the ansatz it can lose universality. This can happen if the gates involved do not generate the full set of unitaries compatible with the symmetry. In the context of quantum chemistry and fermi-to-qubit mappings, 3-local unitaries are necessary to generate all possible symmetry-preserving transformations~\cite{arrazola2021universal,oszmaniec2017universal}. If one would naively symmetrize an ansatz built from the customary single- and two-qubit gates one would loose universality. 

Third, the circuit depth could increase dramatically under symmetrization. This would happen when the symmetry induces interactions which are non-local with respect to the underlying hardware and which have to be realized by extensive SWAP-chains. 

We also want to note that making the ansatz equivariant is not the only way to make use of symmetries. One way to do so is to use \emph{penalty terms}, which are additional parts of the cost function that increase the cost of states that have low symmetry~\cite{kuroiwa2021penalty,ryabinkin2018symmetry,yen2019exact}. This is comparable to the use of regularization terms in classical machine learning. Penalty terms can be constructed for arbitrary symmetry representations and have the charm that they act through the \emph{optimizer} instead of the ansatz, which means they can be generically included to partially alleviate problems with uninformed ansätze.

\section{Invariant re-uploading models}\label{sec:invariant_models}

In this section, we want to summarize the construction of invariant re-uploading models. We repeat the definition of a generic data re-uploading model in our sense. We prepare an ansatz state vector consisting of repeated applications of trainable blocks $\{ W_i(\ttheta) \}$ and a data-encoding unitary $U(\xx)$
as
\begin{align}
    \ket{\psi(\ttheta, \xx)} =  W_L(\ttheta) U(\xx) \dots W_1(\ttheta) U(\xx) W_0(\ttheta) \ket{\psi_0}.
\end{align}
A prediction is then obtained from the expectation value 
\begin{align}
    y(\xx) = \langle \psi(\ttheta, \xx) | O | \psi(\ttheta, \xx) \rangle
\end{align}
of an observable. We wish to construct the quantum learning model in a way that the prediction is invariant under the action of a symmetry group $\calS$ acting through a representation $V_s$ for $s \in \calS$ in the sense that
\begin{align}
    y(V_s[\xx]) = y(\xx) \text{ for all } \xx, s \in \calS.
\end{align}
As we have argued in Sec.~\ref{sec:symmetric_structures}, we can do so for quantum learning models if the data embedding $U(\xx)$ is \emph{equivariant} with respect to the symmetry in the sense that it \emph{induces} a unitary representation of the symmetry group on the Hilbert space according to
\begin{align}
    U(V_s[\xx]) = U_s U(\xx) U_s^{\dagger} \text{ for all } s\in \calS.
\end{align}
If the data embedding is not equivariant we have little hope to build invariant quantum learning models as the symmetry is not meaningfully represented on the level of the Hilbert space. Under an equivariant embedding, the parts of the re-uploading model transform like
\begin{align}
     U(V_s[\xx]) W_i(\ttheta) U(V_s[\xx])  =  U_s U(\xx) U_s^{\dagger} W_i(\ttheta) U_s U(\xx) U_s^{\dagger}.  
\end{align}
We can make this construction equivariant by enforcing that the trainable block $W_i(\ttheta)$ is \emph{equivariant} with respect to the unitary representation $U_s$ of $\calS$, mathematically formulated as 
\begin{align}\label{eqn:equivariance_trainable_block}
    [W_i(\ttheta), U_s] = 0 \text{ for all } \ttheta, s\in \calS.
\end{align}
As argued in Sec.~\ref{sec:gateSymmetrization}, the representation $U_s$ gives us a direct recipe to enforce equivariance of the trainable blocks. This is because the \emph{twirling formula} 
\begin{align}
    \calT_U[X] = \frac{1}{|\calS|} \sum_{s \in \calS} U_s X U_s^{\dagger}
\end{align}
is a projector onto all operators that commute with the symmetry representation (a projection onto the \emph{commutant}) and are thus \emph{equivariant}. We use the fact that ansätze are constructed from parametrized gates with generators a gateset $\calG$ to define an \emph{equivariant gateset} $\calT_U[\calG]$, which consists of the twirled generators of the standard gateset. The equivariant gateset now contains building blocks that can be freely combined to construct equivariant trainable blocks. 
Combining data embeddings equivariant with respect to the data symmetry with trainable blocks that are equivariant with respect to the induced symmetry on the Hilbert space gives rise to an \emph{equivariant circuit}
\begin{align}
    U(\ttheta, V_s[\xx]) = U_s U(\ttheta, \xx) U_s^{\dagger} \text{ for all } \ttheta, \xx, s \in \calS.
\end{align}

We have so far shown how to construct equivariant circuits. To reach invariance of the final prediction, we also need an \emph{invariant} initial state vector
\begin{align}
    U_s \ket{\psi_0} = \ket{\psi_0} \text{ for all } s \in \calS.
\end{align}
As $U_{s} = U_{s^{-1}}^{\dagger}$ this equivalently implies that 
\begin{align}
    U_s^{\dagger} \ket{\psi_0} = \ket{\psi_0} \text{ for all } s \in \calS
\end{align}
If we now apply an \emph{equivariant circuit} to an \emph{invariant initial state}, we obtain an \emph{equivariant ansatz}
\begin{align}
    \ket{\psi(\ttheta, V_s[\xx])} &= U(\ttheta, V_s[\xx]) \ket{\psi_0} \\
     &= U_s U(\ttheta, \xx)U_s^{\dagger} \ket{\psi_0} \\
     &= U_s U(\ttheta, \xx) \ket{\psi_0} \\
     &= U_s \ket{\psi(\ttheta, \xx)}\text{ for all } \ttheta, \xx, s \in \calS.
\end{align}
Finally, if we combine an equivariant ansatz with an \emph{invariant observable}
\begin{align}\label{eqn:invariant_observable}
    U_s O U_s^{\dagger} = O \text{ for all }s \in \calS,
\end{align}
we obtain an \emph{invariant re-uploading model}
\begin{align}
    y(V_s[\xx]) &= \langle {\psi(\ttheta, V_s[\xx])} | O | {\psi(\ttheta, V_s[\xx])} \rangle\\
    &= \langle {\psi(\ttheta, \xx)} | U_s^{\dagger} O U_s | {\psi(\ttheta, \xx)} \rangle\\
    &= \langle {\psi(\ttheta, \xx)} | O | {\psi(\ttheta, \xx)} \rangle\\
    &= y(\xx)\text{ for all } \ttheta, \xx, s \in \calS.
\end{align}

Note that the conditions for invariance of the observable in Eq.~\eqref{eqn:invariant_observable} the condition for the equivariance of the trainable blocks in Eq.~\eqref{eqn:equivariance_trainable_block} are actually the same if spelled out. The reason for this apparent contradiction is that both objects are operators, but they do different things: the trainable block performs a transformation of states, but the initial state and the final observable act like objects that are transformed. The picture becomes more clear if we look at the operations in terms of their action on density matrices, \ie if we understand them as quantum channels. Then, equivariance of a trainable block manifests as
\begin{align}
    W_i(\ttheta) U_s \rho U_s^{\dagger} W_i(\ttheta)^{\dagger} = U_s W_i(\ttheta) \rho W_i(\ttheta)^{\dagger} U_s^{\dagger},
\end{align}
whereas invariance of an observable is still realized as in Eq.~\eqref{eqn:invariant_observable}. Both pictures come together in the fact that conjugation with an \emph{invariant} object (in our case the unitary representing the trainable block) constitutes an \emph{equivariant} transformation. 

The construction outlined above can be viewed as the \enquote{adjoint} or \enquote{two-sided} version of the geometric deep learning blueprint where we have to start from an invariant object, the initial state, apply equivariant operations and then evaluate the expectation value again on an invariant object. 


\section{Numerical experiments}\label{sec:numerics}

In this section, we report on the results of numerical experiments we undertook to compare the performance of invariant learning models with their non-invariant counterparts on two selected toy problems with non-trivial symmetry. We observe that invariant learning models under-perform on the training data but have a much better generalization performance than learning models that do not respect the symmetry considerations. To ensure that the symmetrization procedure proposed in Sec.~\ref{sec:gateSymmetrization} helps generically and we did not cherry-pick a working edge-case, we also performed simulations where we randomized over the layouts of the trainable parts of our learning models that confirm our main numerics and show a generic advantage in generalization.

We have further investigated the performance of equivariant ansätze for ground state problems at the hands of the transverse-Field Ising model, the Heisenberg model and a longitudinal-transverse-field Ising model that shares the same symmetries as our learning toy problems. For these problems, and in the ansatz constructions we have studied, equivariant ansätze produce better ground state approximations on average while needing fewer iterations to converge. We further show that equivariant ansätze can mitigate the barren plateaus problem. We note, however, that the applicability picture is not as clear as in the quantum machine learning application and that equivariant ansätze can also have downsides which we discuss in detail.

\subsection*{Tic-tac-toe}
To showcase the methods outlined above we start by considering a simple training task based on the well-known \emph{game of tic-tac-toe}. The goal will be to train a variational quantum learning model to classify games into the categories \enquote{cross won} ($\times$), \enquote{circle won} ($\circ$) or \enquote{draw} ($-$). This classification problem has a non-trivial symmetry group, as rotating the board and reflecting the board about an axis does not change the outcome. It also has the advantage that it can be realized with equivariant embeddings on a modest number of nine qubits. The symmetries of the tic-tac-toe game are depicted on the left of Fig.~\ref{fig:ttt}.

The symmetry group of this learning task is given by the \emph{dihedral group} $D_4$, which is equivalent to all operations that map a square to itself. The dihedral group can be generated by a counter-clockwise rotation of 90 degrees and a flip about the vertical axis through the center. The group has order $8$, so has a total of $8$ elements. The group induces equivalence classes of fields of the tic-tac-toe board, as corners of the board will always be mapped to corners, edges to edges and the center will always stay the same. These three equivalence classes will also be mirrored in the equivariant gateset. 

We note that the task of labeling games of tic-tac-toe is easily solved by a classical deterministic algorithm, and in fact, only a finite number of games are even possible. The goal of this numerical example is to showcase a possible end-to-end implementation of our symmetrization procedure and to compare equivariant with standard gatesets. Furthermore, in some of the following sections, we discuss how this example can be taken as a paradigm to tackle other, potentially more useful problems.

\begin{figure*}
    \centering
    \includegraphics{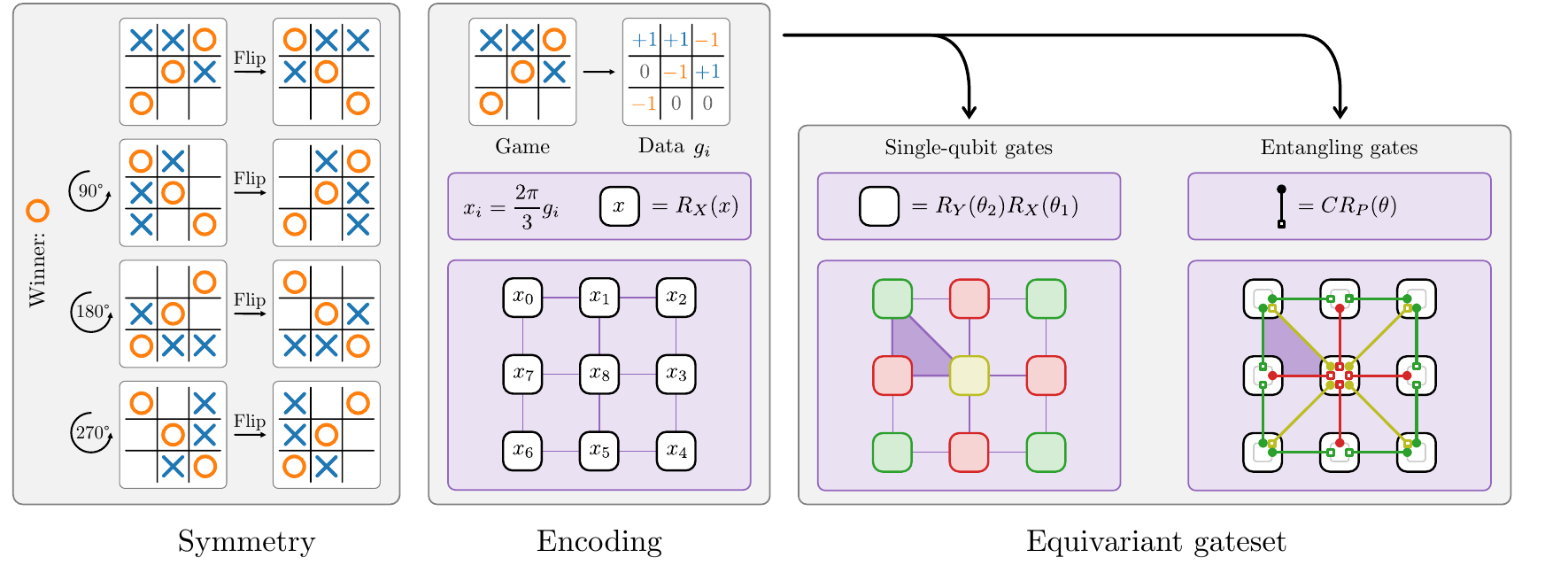}
    \caption{The winner of a game of tic-tac-toe is independent of rotations and flips of the board, which means its symmetry group is given by the dihedral group $D_4$, as visualized on the left side of the figure. We use a dataset of tic-tac-toe games encoded via Pauli-$X$ rotations on nine qubits, as visualized in the second column of the plot. The three different values of a field are equidistantly encoded by using a multiple of ${2\pi}/{3}$ for the rotation angle. Using our symmetrization procedure, we obtain an equivariant gateset made up of single-qubit gates acting identically on the corners (green), edges (red) and the middle (yellow) of the board. The two-qubit gates are given by controlled rotations between corners and edges (green), edges and the middle (red) and the middle and the corners (yellow) where the control qubit was always listed first. Due to the symmetry, it is sufficient to specify the gates on one neighboring trio of corner edge and middle, as visualized by the dark purple triangle.}
    \label{fig:ttt}
\end{figure*}

\subsubsection*{Dataset}

The first process is mapping a game of tic-tac-toe to classical data. As shown in the second column of Fig.~\ref{fig:ttt}, we do so by mapping the nine fields of the board to elements of a vector $\gg$ where $+1$ represents a cross, $-1$ a circle and $0$ an empty field. The labels of the game are encoded in a one-hot fashion in a vector $\yy = (y_{\circ}, y_{-}, y_{\times})$ where a value of $+1$ is assigned to the correct label and $-1$ to the other two entries. 
The number of distinct tic-tac-toe games is sufficiently small so that we were able to generate all possible valid games. In our dataset, we also allow for unfinished games, which are labeled as \enquote{draw}. 
The training and test data sets were then constructed by choosing a subset of all possible games of tic-tac-toe at random, with the constraint that each outcome is equally represented. 

\subsubsection*{Learning model}

To address the tic-tac-toe learning task, we will make use of a data re-uploading model with an equivariant embedding as described in Secs.~\ref{sec:preliminaries} and \ref{sec:symmetric_structures}. 
The equivariant embedding is constructed by encoding the different numerical values that represent a game via a Pauli-$X$ rotation on separate qubits that we view in a planar grid. To distribute the three data features equidistantly, we use a multiple of $2\pi/3$ for the rotation angle, again as shown in the second column of Fig.~\ref{fig:ttt}.  

The equivariance of the embedding ensures that symmetry transformations are realized by unitary conjugation. For example, a reflection along the vertical axis is implemented by $\SWAP_{02}\SWAP_{73}\SWAP_{64}$ in the numbering of Fig.~\ref{fig:ttt}. The advantage of permutation-type symmetries on the level of qubits is that they can easily be understood on a visual level. 
The qubits lie in equivalence classes under the symmetry operation. Single-qubit gates then have to act on all qubits of the same equivalence class equally. This is how we get the type of equivariant single-qubit layer that is depicted in the third column of Fig.~\ref{fig:ttt} where single-qubit unitaries applied that share the same parameters when acting on corners (c), on edges (e) or on the middle (m). The same reasoning can also be applied to two-qubit gates. An entangling operation that connects, for example, a corner with the edge next to it has to act in the same way on all pairs of neighboring corners and edges. This is how we obtain the equivariant layers of entangling gates used for our learning models, which perform controlled rotations from corners to neighboring edges (o), edges to the middle (i) and from the middle to the corners (d), as depicted in the fourth column of Fig.~\ref{fig:ttt}. We chose $CR_Y$ for parametrized rotations.

The learning model starts with all qubits initialized in the $\ket{0}$ state vector, which is invariant under the problem's symmetry. Then a number of layers are applied, each made up of one data encoding followed by a sequence of parametrized layers. The default parametrized layer was chosen to be \enquote{cemoid}, which corresponds to one application of the single-qubit gates followed by the entangling gates, both visible in Fig.~\ref{fig:ttt}. The prediction of a label is obtained in a one-hot-encoding by measuring the expectation values of three invariant observables
\begin{align}
    O_{\circ} &= \frac{1}{4}\sum_{i \in \text{corners}} Z_i= \frac{1}{4}[Z_0 + Z_2 + Z_4 + Z_6],\\
    O_{-} &= Z_{\text{middle}} = Z_8, \\
    O_{\times} &= \frac{1}{4}\sum_{i \in \text{edges}} Z_i = \frac{1}{4}[Z_1 + Z_3 + Z_5 + Z_7]
\end{align} 
as $\hat{\yy} = (\langle O_{\circ} \rangle, \langle O_{-}\rangle, \langle O_{\times} \rangle)$.  A prediction for a given data point is obtained by selecting the class for which the observed expectation value is the largest.

\subsubsection*{Training}

We train the learning model using an $l_2$-loss function, which for a set of games with associated one-hot label vectors $\calD = \{ (\gg, \yy) \}$ is given by
\begin{align}
    \calL(\calD)  = \frac{1}{|\calD|}\sum_{(\gg, \yy) \in \calD} \lVert \hat{\yy}(\gg) - \yy \rVert_2^2.
\end{align}
We ran optimizations with 100 epochs, each consisting of 30 steps. At each step, the gradient was computed using 15 data points representing a tic-tac-toe game. The size of the training set is then the product of these two numbers $15\times 30 = 450$. The training set is shuffled after each epoch is completed.
The test set was composed of 600 games chosen randomly for each training run with the same constraint as above but kept fixed throughout the run. 
The above hyper-parameters were chosen empirically. 

The quantum learning model was implemented with the PennyLane library~\cite{bergholm2020pennylane} for quantum machine learning. Using the PyTorch interface provided by PennyLane, we trained the PQC by gradient descent as implemented by the Adam PyTorch optimizer~\cite{NEURIPS2019_9015}.

\subsubsection*{Results}

In all our numerics, we compare invariant with non-invariant models where the parameter sharing indicated in Fig.~\ref{fig:ttt} is not imposed. The non-invariant models therefore have more independent parameters and hence a higher expressivity.
To evaluate the performance of the models, we record the classification accuracy for the training and the test set.
The general trend we observe is that invariant models built from equivariant circuits achieve the same or lower accuracy on the training set, but consistently higher accuracy on the test set which indicates their better generalization capabilities.

We first verify that this is the case for circuits of different sizes.
Fig.~\ref{fig:lp_sweep_ttt} compares the results of invariant and non-invariant models composed of $l$ layers, each of which consists of one data encoding followed by $p$ independent repetitions of the layout \enquote{cemoid}.
We refer to the different architectures with the pair $(l,p)$, and we sweep over the range $(l,p)\in\{1,\dots,5\}^2$.
Due to limited computational resources, we were unable to record the results for the values {$(l,p)\in\{(4,4), (4,5), (5,3), (5,4), (5,5)\}$}.
In this experiment, the parametrized entangling gates were chosen to be $CR_Y(\theta)$, but similar experiments with different controlled Pauli rotations produce comparable results.

The difference in accuracy on the training set should not come across as odd, since invariant models only express a subset of the output mappings of the non-invariant ones.
For non-invariant models, high training accuracy and low test accuracy are clear witnesses of the overfitting regime.
On the contrary, for invariant models we see that the price we pay by lowering the expressivity returns a very similar performance on both previously seen and unseen data, hinting at the sweet spot in the bias-variance trade-off.
Formally, we say the empirical generalization gap of invariant models is much smaller than that of non-invariant ones, which confirms our expectations.

As a sanity check, and to make sure we did not cherry-pick a working example, we studied the performance of other variational layouts different from \enquote{cemoid}. We did so by fixing the number of layers to $l=3$, where after encoding the data with \enquote{t} we take three consecutive random permutations of \enquote{cemoid} without contiguous repetitions. An example of such a layout is \enquote{tdomiececiodmdmoiec}, which we then repeat three times. In total, 20 different layouts were generated and simulated.
The summary of results in Fig.~\ref{fig:randomization_ttt} reveals that the test performance of the invariant models is consistently higher across different layer layouts, with few exceptions. This confirms our expectations that making learning models invariant helps more or less generically to enhance generalization performance.

\begin{figure}
    \centering
    \includegraphics{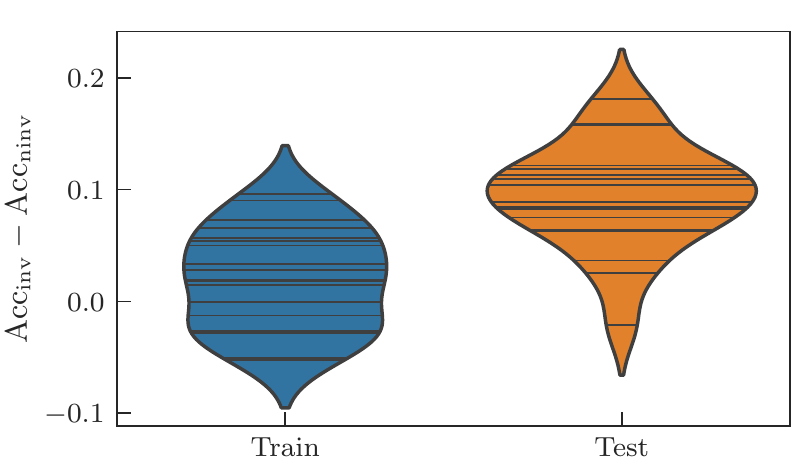}
    \caption{Violin plot of the difference in achieved mean accuracies between the invariant and non-invariant model in a sweep over different values of $l$ and $p$ as detailed in the main text. Positive values correspond to invariant models outperforming non-invariant ones. The mean was averaged over ten runs for each combination of $l$ and $p$. We see that the mean accuracy on the training data is more or less the same for both invariant and non-invariant models, but that the invariant models clearly outperform the non-invariant ones on the test data. The full training graphs can be found in the appendix in Fig.~\ref{fig:lp_sweep_full_ttt}.}
    \label{fig:lp_sweep_ttt}
\end{figure}

\begin{figure}
    \centering
    \includegraphics{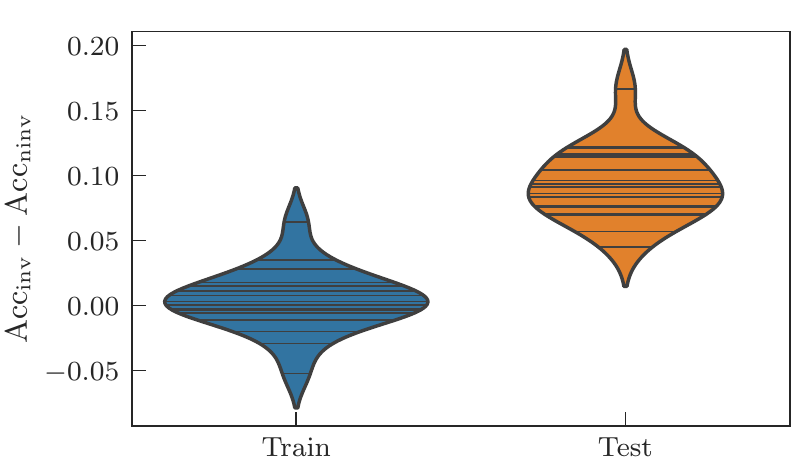}
    \caption{Violin plot of the difference in achieved accuracies between the invariant and non-invariant model for 20 random parametrizations of the trainable blocks as detailed in the main text. Positive values correspond to invariant models outperforming non-invariant ones. The values for every parametrization were averaged over ten repetitions with random parameter initializations. We see that the mean accuracy on the training data is more or less the same for both invariant and non-invariant models, but that the invariant models again clearly outperform the non-invariant ones on the test data.}
    \label{fig:randomization_ttt}
\end{figure}

\subsection*{Classifying autonomous vehicle scenarios}
The tic-tac-toe task provides an intuitive example of how we can exploit symmetry in learning, but it does not connect to a relevant real-world learning scenario. In this section, we want to outline how we can use the intuition developed in the tic-tac-toe example and apply it to a toy model of a task that is of actual relevance in the automotive industry and shows a clear way of connecting to a real-world scenario.

Autonomous driving is a future-oriented field of the automotive industry, for which two main challenges interact. First, the vehicle must be able to recognize and automatically evaluate its surroundings to deduce possible driving maneuvers to reach a goal. Second, the automated evaluation performed by the car requires verification and testing along the development cycle of the car. To meet these challenges, scenario-based development is state of the art \cite{2018arXiv180108598M}. We define a scenario as a concatenation of a scene (snapshot of the surrounding) and actions (destination goals and values). Therefore a scenario is a specific description of a driving situation, taking into account dynamic and static components which are determined by the sensor systems of the vehicle, map data, and others \cite{unifiedontology}. In the development process of the vehicle different pre-defined scenarios are classified with respect to their safety relevance according to a criticality metric. This classification will then be used to evaluate the requirements for testing, in the field or in simulation. 
The Pegasus Project \cite{2021arXiv210409097S} classifies the scenarios into different levels. In this publication, we concentrate on street level 1, which includes geometry and topology of the streets, in order to demonstrate our concept. Driving-maneuver classification is investigated using classical machine learning tools among other approaches \cite{7995704,8166755}. We again want to emphasize that, for the purpose of this paper, the original classification task is reduced to a simple version that can be simulated with the computational resources available to us. It is however clear how larger resources in terms of quantum computational hardware would allow us to address parts of the actual question as soon as it becomes available.

\subsubsection*{Dataset}
In the following, we derive driving scenarios at different street intersections, deduce the geometric symmetries and build a classification of the possible actions according to their safety relevance. Since relevant scenarios include behavior at intersections, traffic circles, and traffic jams \cite{schwab2019car}, we consider the following simplified street situations, using a $3\times3$ grid. 
Each tile of this grid can be either part of a road or be impassable, and a car is placed on one of the road tiles. With this encoding, we are able to represent different geometries from a straight road to a more complicated intersection. The symmetries of these scenarios are a subset of the Tic-Tac-To case as a scenario can always be rotated by 90°(see the right panel of Fig.~\ref{fig:pdg_difficulties}), but not necessarily mirrored in all instances as left turns have a different difficulty than right turns. 
We translate this into a $3\times3$ array $\xx$, where a road tile is encoded with 1, an impassable tile with -1, the car with -1/3, and the orientation of the car with 1/3. An example of this can be seen in the middle panel of Fig.~\ref{fig:pdg_difficulties}. 

We rate the level of difficulty $y$ with regards to a simplified criticality metric as follows: Forward: 0; Right or left curve: 0.2; Forward and right (T-crossing, intersection of three streets): 0.4; Forward and left (T-crossing, intersection of three streets): 0.6; Left and right (T-crossing, intersection of three streets): 0.8; Forward, left and right (X-crossing, intersection of four streets): 1. 
Examples of the possible difficulties are visualized in the left panel of Fig.~\ref{fig:pdg_difficulties} alongside a representation of the data encoding and a demonstration of the symmetry.

The different scenarios $\xx$ were generated by first hand placing various road layouts (T-intersection, left curve, X-intersections...), generating all their images under rotations and reflections and then putting the car on every possible road tile. Additionally, every possible orientation of a placed car is iterated. This process creates the total set of possible scenarios and their associated difficulty $\calD = \{(\xx, y)\}$.

As for the tic-tac-toe games, training and test data sets were constructed by choosing a subset of scenarios at random, with the constraint that each difficulty level is equally represented.

\begin{figure*}
    \centering
    \includegraphics{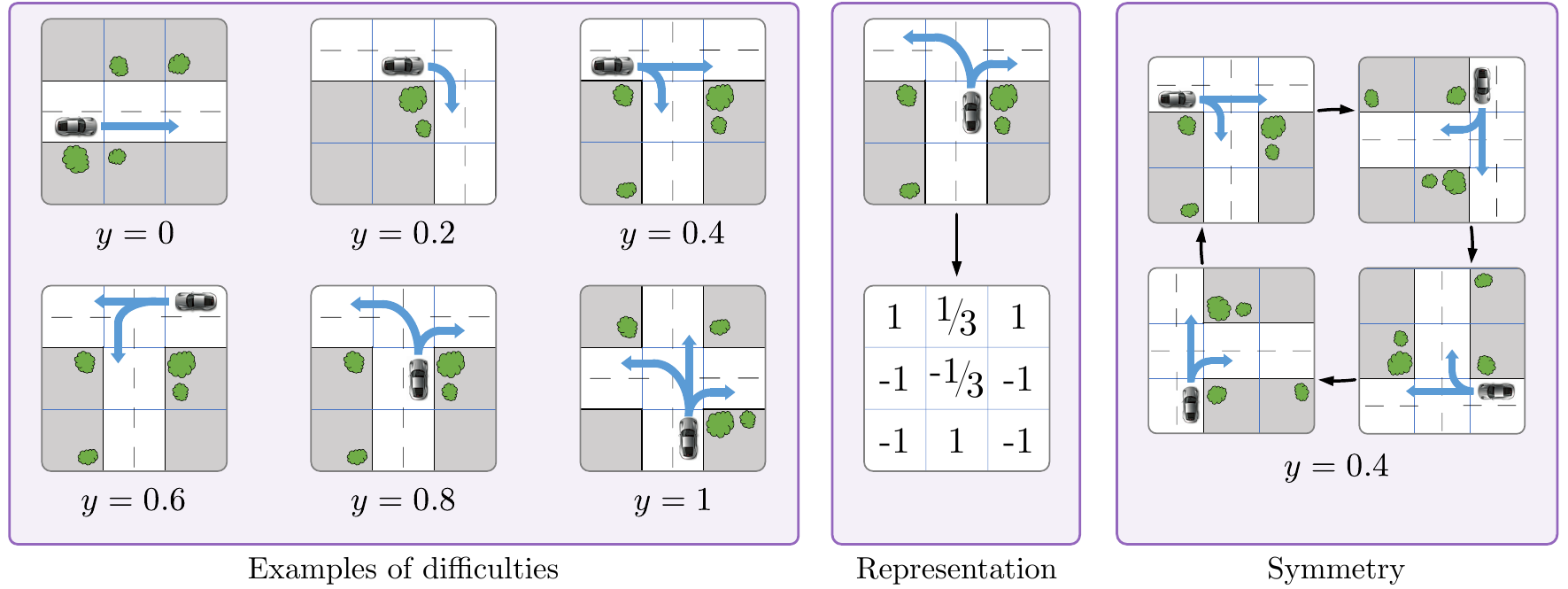}
    \caption{Examples of the various difficulties encountered in the autonomous driving toy model (left), a demonstration of the data representation (middle) and a demonstration of the rotational symmetry of the model (right).}
    \label{fig:pdg_difficulties}
\end{figure*}

\subsubsection*{Learning model}
The general circuit construction is essentially the same as in the tic-tac-toe case and follows the data re-uploading models as described in Secs.~\ref{sec:preliminaries} and \ref{sec:symmetric_structures}.
The different numerical values that represent a scene are encoded via a Pauli-X rotation on the separate qubits where the rotation angle is given by a multiple $2\pi/3$ of the array value as in the tic-tac-toe case.
The key difference to the tic-tac-toe case is the missing mirror symmetry, rendering the actual symmetry group to $\bbZ_4$ instead of $D_4$. This only affects the outer layer (o) and splits it into two distinct sub-layers: one where controlled operations are performed clockwise with shared parameters and one where controlled operations are performed counter-clockwise with shared parameters.

As in the tic-tac-toe case, all qubits are initialized in the $\ket{0}$ state and layers of data-encoding followed by parametrized gates are applied. The default topology for the parametrized blocks was again chosen to be \enquote{cemoid}, where it is understood that the outer layer (o) splits into two sub-layers as explained above. The model's prediction $\hat{y}$ of difficulty is obtained by measuring and normalizing the Pauli-Z expectation value of the middle qubit $Z_8$
\begin{equation}
    \hat{y} = \frac{\langle Z_{\text{middle}}\rangle + 1}{2}= \frac{\langle Z_8\rangle + 1}{2}.
\end{equation}
A hard prediction for a given scenario $\xx$ is obtained by rounding $\hat{y}$ to the nearest difficulty in
$\{0, 0.2, 0.4, 0.6, 0.8, 1\}$.

\subsubsection*{Training}
The learning model is trained using an $l_2$-loss function which is, for a set of games $\calD = \{ (\xx, y) \}$, given by
\begin{align}
    \calL(\calD)  = \frac{1}{|\calD|}\sum_{(\xx, y) \in \calD} ( \hat{y}(\xx) - y )^2.
\end{align}
No epochs were used in this case and optimizations were run with 30 steps. 
At each step, the gradient was computed for the same 60 scenes. As the limited grid size of $3\times3$ only allowed four difficulty 1 [X-Crossing] scenarios, random copies were created to ensure an equal distribution (ten games of each difficulty). After each step, the accuracy on the training data and the test data, which consisted of 130 unique games, was evaluated by calculating the fraction of correctly classified inputs.
The above hyper-parameters were again chosen empirically. 
The numerical experiments have been performed using the PennyLane library~\cite{bergholm2020pennylane}. For the two-qubit gates, only the $CR_Z$-gate has been implemented. All optimizations have been executed using the PyTorch L-BFGS-optimizer~\cite{NEURIPS2019_9015}.

\subsubsection*{Results}

\begin{figure}
    \centering
    \includegraphics{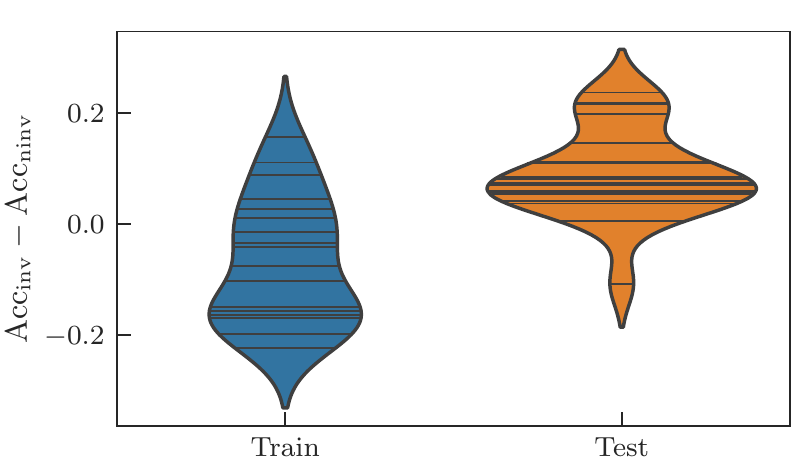}
    \caption{Violin plot of the difference in achieved mean accuracies between the invariant and non-invariant model in a sweep over different values of $l$ and $p$ as detailed in the main text. Positive values correspond to invariant models outperforming non-invariant ones. The mean was averaged over ten runs for each combination of $l$ and $p$. We see that the mean accuracy on the training data is more or less the same for both invariant and non-invariant models, but that the invariant models clearly outperform the non-invariant ones on the test data. The full training graphs can be found in the appendix in Fig.~\ref{fig:lp_sweep_full_porsche}.}
    \label{fig:lp_sweep_porsche}
\end{figure}

\begin{figure}
    \centering
    \includegraphics{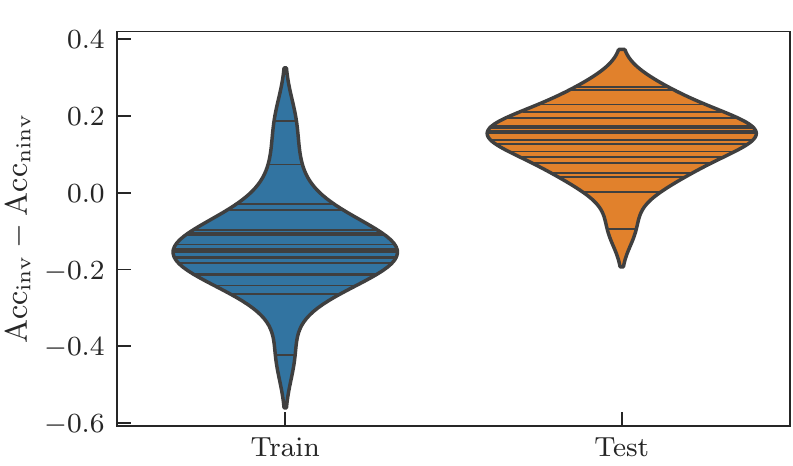}
    \caption{Violin plot of the difference in achieved accuracies between the invariant and non-invariant model for 20 random parametrizations of the trainable blocks as detailed in the main text. Positive values correspond to invariant models outperforming non-invariant ones. For every parametrization, ten random initializations were performed and then averaged. We see that the mean accuracy on the training data is more or less the same for both invariant and non-invariant models, but that the invariant models clearly outperform the non-invariant ones on the test data.}
    \label{fig:randomization_porsche}
\end{figure}

We repeat the experiments we designed for the tic-tac-toe learning task, now with this different dataset.
As expected, we observe the same general trend.
We repeat that we compare an invariant model with the non-invariant model obtained from the invariant one where the parameter sharing indicated in Fig.~\ref{fig:ttt} is not imposed. The non-invariant models therefore have more independent parameters and hence a higher expressivity.
To evaluate the performance of the models, we record the classification accuracy for the training and the test set.

We again study different hyperparameters for the architecture of the learning model. Recall that we use $l$ layers consisting of a data encoding followed by $p$ repetitions of the atom \enquote{cemoid}.
We sweep over the range $(l,p)\in\{1,\dots,5\}^2$ except for the values $(4,4), (4,5), (5,3), (5,4)$, and $(5,5)$.
As can be seen in Fig.~\ref{fig:lp_sweep_porsche}, most invariant learning models perform worse than non-invariant models on the training set and better on the test set.
A straightforward explanation would be the position the different models take in the bias-variance trade-off: overfitting regime for non-invariant models, and sweet-spot for invariant models.
The full data found in Fig.~\ref{fig:lp_sweep_full_porsche} in the appendix further supports this intuition as it shows a drop in accuracy for high $(l,p)$ values, where the expressivity of both models is found empirically to be too high for this task.
This experiment crystallizes in the statement: invariant models generalize better than non-invariant models.

We also repeat the randomization experiment over the specific spelling of the trainable parts, with the results being reported in Fig.~\ref{fig:randomization_porsche}.
We make circuits with a random layer repeated $3$ times.
Random layers start with a data encoding followed by three consecutive random permutations of \enquote{cemoid}, such that there are no repeated letters adjacent.
Again, in almost every case the invariant models outmatch the non-invariant ones.

\subsection*{Variational quantum eigensolvers}

The gate symmetrization procedure proposed in this work only necessitates the existence of a unitary symmetry representation on the level of the Hilbert space. It can therefore also be applied to ground state problems where conserved quantities of the Hamiltonian yield representations of symmetries. This case is paradigmatically treated through a variational algorithm using the variational quantum eigensolver. In the following, we will conduct some numerical experiments that showcase the advantages and disadvantages of generic symmetrization in this context.

\subsubsection*{Transverse-field Ising model}
We will first consider the \emph{transverse-field Ising model (TFIM)}~\cite{Franchini_2017} as a paradigmatic example. The TFIM Hamiltonian with periodic boundary conditions on $N$ spins is given by
\begin{align}
    {H}_{\text{TFIM}}=-\sum_{i=1}^{N}Z_i Z_{i+1} - g\sum_{i=1}^{N}X_i,
\end{align}
where we consider a transverse field strength $g>0$.

The TFIM Hamiltonian has a $\mathbb{Z}_2$ symmetry as it commutes with the parity operator
\begin{align}
    P = \prod_{i=1}^N X_i
\end{align}
The unitary representation is then given by $U_s = P^s$ for $s \in \bbZ_2$. The eigenvalues of the parity operator are either $+1$ or $-1$. For $g \rightarrow \infty$ the ground state is given by $\ket{+}^{\otimes N}$ which has a parity of $+1$. 
Using the \emph{adiabatic theorem} and the fact that for finite system size the ground state energy is not degenerate we can conclude that the parity of the ground state is the same for each $g>0$. Therefore, if we want to force our ansatz state to encode this symmetry, we require that
\begin{align}
    P \ket{\psi({\ttheta})}= + \ket{\psi({\ttheta})}.
\end{align}
Of course, it is not necessary for our ansatz to respect this property for all values of the parameters as long as it finds the correct ground state, but in many cases it can be beneficial to restrict the expressivity of the ansatz into a relevant part of the Hilbert space. If we do this, however, we have to assure that the ansatz does produce a state that is in the same symmetry sector as the true ground state. This can be assured by taking an initial state that has the correct symmetry, for example $\ket{\psi_0}=\ket{+}^{\otimes N}$, and then only performing equivariant gates that can be obtained from the symmetrization procedure of Sec.~\ref{sec:gateSymmetrization}. We note again that one has to be attentive to the fact that symmetrization has to be executed with care as it can trivialize certain generators.

In our numerical experiments, we use the QAOA ansatz~\cite{farhi2014quantum} which can be easily seen to be equivariant with respect to the parity symmetry,
\begin{align}
\begin{split}
&\ket{\psi_{\text{QAOA}}\left(\boldsymbol{\beta},\ggamma\right)} = \\
&\qquad \prod_{m=1}^{p}\prod_{i=1}^{N}e^{-i \beta_{m} X_i}\prod_{i=1}^{N}e^{-i \gamma_{m} Z_i Z_{i+1}}\ket{\psi_0}.
\end{split}
\end{align}
Here, $p$ is the number of QAOA layers.
We compare it to a variant of the QAOA ansatz which we denote as QAOA' where we add an additional mixer term involving Pauli-$Y$ rotations increases the 
expressivity as
\begin{align}
\begin{split}
&\ket{\psi_\text{QAOA'}(\aalpha, \bbeta, \ggamma)} = \\
&\qquad \prod_{m=1}^{p}\prod_{i=1}^{N}e^{-i \alpha_{m} Y_i}\prod_{i=1}^{N}e^{-i \beta_{m} X_i}\prod_{i=1}^{N}e^{-i \gamma_{m} Z_i Z_{i+1}}\ket{\psi_0},
\end{split}
\end{align}
As $P Y_i P = - Y_i$ the gate symmetrization procedure of Sec.~\ref{sec:gateSymmetrization} will remove these Pauli-$Y$ rotations and will yield the standard QAOA ansatz, which has been widely studied and shown \cite{Mbeng,Ho_2019,Wierichs_2020,Wiersema_2020,Wierichs_2020} to provide a faithful ground state approximation for the TFIM when $p\ge N / 2$, while for $p< N /2$ the QAOA ansatz can only reach a variational energy which is above the ground state energy \cite{Ho_2019} due to light-cone arguments \cite{Mbeng}. This behavior can also be observed in our experiments.

\begin{figure}
\centering
\includegraphics{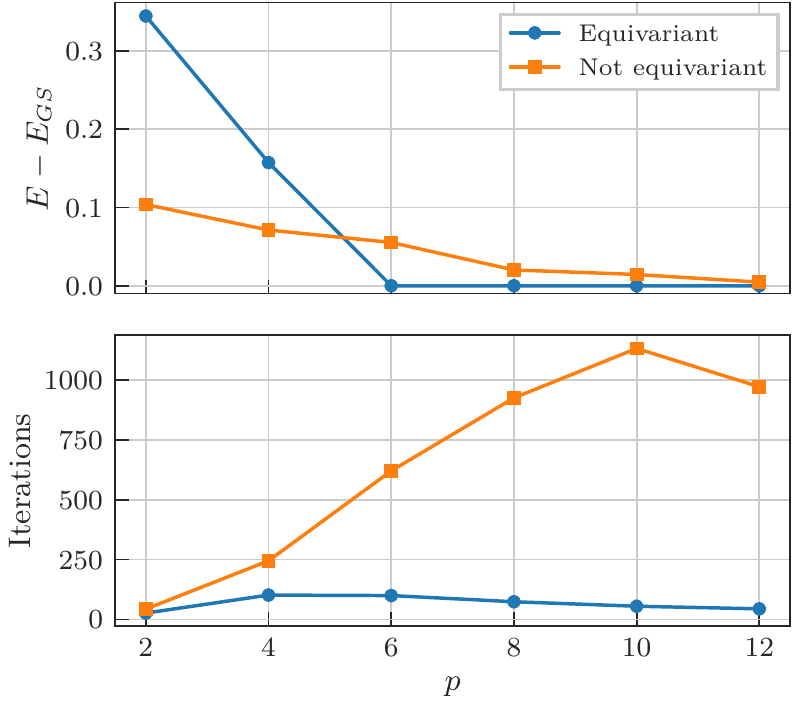}
\caption{Comparison of the equivariant QAOA ansatz with the non-equivariant QAOA' ansatz for the TFIM on $N=10$ spins with field $g=1$. For every number of layers $p$, we perform 20 experiments with random initializations of the parameters and plot the mean achieved energy (top) and the number of iterations necessary to reach it (bottom). We see that for $p\ge N/2$ the equivariant ansatz performs better in both figures of merit.}
\label{fig:TFIM_example}
\end{figure}

We compare the performance of the two ansätze using the TFIM with $N=10$ spins and different values for the number of QAOA layers $p$. We optimize using the L-BFGS optimizer until convergence is reached. We perform the optimization for 20 random initializations of the circuit parameters and average the results to arrive at the statistics shown in Fig.~\ref{fig:TFIM_example}. 
For all values of $p$, the QAOA ansatz needs fewer iterations to converge. For small $p < N/2$, we observe that the QAOA ansatz does not converge to a good approximation of the ground state which is in line with the aforementioned previous findings. In this regime, it is outperformed by the non-equivariant ansatz, highlighting the trade-off between expressivity and equivariance. If the circuit depth $p\ge {N}/{2}$ is large enough, the picture reverses and the equivariant QAOA ansatz reliably reaches the ground state whereas the non-equivariant QAOA' ansatz converges to an energy above the ground state on average. 

\subsubsection*{Heisenberg model}
Another model that has a continuous symmetry group is given by the Heisenberg model~\cite{Franchini_2017} with periodic boundary conditions on an even number $N$ of spins captured by the Hamiltonian
\begin{align}
    {H}_{\text{Heis}}=\sum_{i=1}^{N}X_i X_{i+1}+Y_i Y_{i+1}+Z_i Z_{i+1}.
\end{align}
We can understand this Hamiltonian as the alignment of two neighboring spins. Quite logically, if we rotate all spins simultaneously, the \emph{relative} alignment will not change. The symmetry group of the model is thus $\SU(2)$ represented by
\begin{align}
    U_V = V^{\otimes N} \text{ for } V \in \SU(2).
\end{align}
As in the previous example, we initialize the quantum computer in a state that lies in the right symmetry sector, which in the case of the Heisenberg model is given by all states that have zero total spin, and we choose
\begin{align}
    \ket{\psi_0} = \bigotimes_{k=1}^{N/2} \frac{1}{\sqrt{2}}(\ket{0,1} - \ket{1,0}).
\end{align}
For the Heisenberg model, we start from the non-equivariant ansatz
\begin{align}
\begin{split}
&\ket{\psi(\boldsymbol{\theta})}=\\
&\qquad \prod_{m=1}^{p}\prod_{i=1}^{N}e^{-i \alpha^{(m)} Y_i}e^{-i H_{\rm even}(\bbeta^{(m)})}e^{-i H_{\rm odd}(\ggamma^{(m)})}\ket{\psi_0},
\label{eqn:ansatzHeis}
\end{split}
\end{align}
where $H_{\rm even}(\bbeta_m)$ is an anisotropic Heisenberg Hamiltonian defined on the even lattice sites, 
\ie,
\begin{align}
\begin{split}
    &H_{\rm even}(\bbeta )= \\
    &\qquad \sum_{\substack{i=2\\i \text{ even}} }^{N} \beta_x X_{i-1} X_{i}+\beta_y Y_{i-1} Y_{i}+ \beta_z Z_{i-1} Z_{i}.
\end{split}
\end{align}
The Hamiltonian $H_{\rm odd}(\ggamma)$ is analogously defined but acts on the odd sites.
For the even and odd Hamiltonians we have three parameters each, together with the $\alpha$ parameter controlling the Pauli-$Y$ rotation we have seven parameters per layer, yielding $7p$ parameters in total.
We note that since $H_{\rm even}$ and $H_{\rm odd}$ are both linear combinations of commuting operators, we can decompose the corresponding unitaries using two-qubit gates.

The above ansatz is not equivariant, which means we need to apply the symmetrization procedure. The symmetrization procedure corresponds to a particular instance of a 2-design twirl
\begin{align}
    \calT_U[ X_{1} X_{2}] &= \int \diff \mu(V) \, V^{\otimes N} X_{1} X_{2}(V^{\dagger})^{\otimes N} \\
     &= \int \diff \mu(V) \, V^{\otimes 2} X_{1} X_{2}(V^{\dagger})^{\otimes 2}.
\end{align}
These calculations can be straightforwardly performed using the Weingarten calculus (see, \eg, Ref.~\cite{Pucha_a_2017}). The important fact for us is that any outcome of 2-design twirl will be a sum of identity and SWAP,
\begin{align}
    \calT_U[ X_{1} X_{2}] 
     &=  c_0\bbI +  c_1\SWAP\\
     &=  c_0\bbI + \frac{c_1}{2}(\bbI+X_{1} X_{2}+Y_{1} Y_{2}+Z_{1} Z_{2}),
\end{align}
where we have used the expansion of SWAP into Pauli words. Note that identity terms only generate global phases, which means that the effective symmetrized generator associated to $X_1 X_2$ can be taken to be
\begin{align}
    \calT_U[ X_{1} X_{2}] \sim X_{1} X_{2}+Y_{1} Y_{2}+Z_{1} Z_{2}.
\end{align}
This generator gives rise to the particle-number conserving \emph{Givens rotations}~\cite{givensrotations}.
The additional Pauli-$Y$ rotation in turn is trivialized by the symmetrization procedure
\begin{align}
    \calT_U[Y_{1}] &= \int \diff\mu(V) \,  V Y_{1} V^{\dagger} = \frac{\Tr(Y_{1})}{2}\bbI=0,
\end{align}
where we have exploited the formula associated to first moment operator \cite{Kliesch_2021}.

The same argument holds for the other terms, which yields the 
equivariant ansatz
\begin{align}
\begin{split}
&\ket{\psi(\boldsymbol{\theta})}=\\
&\qquad \prod_{m=1}^{p}\prod_{i=1}^{N}e^{-i \beta^{(m)} H_{\rm even}}e^{-i \gamma^{(m)} H_{\rm odd} }\ket{\psi_0},
\end{split}
\end{align}
where we have chosen
$H_{\rm even}$ and $H_{\rm odd}$ to be the isotropic variants given by $\bbeta =\ggamma = (1,1,1)$. The equivariant ansatz we obtain with the symmetrization procedure has the same form as the \emph{Hamiltonian variational ansatz} used in Ref.~\cite{Wiersema_2020,Ho_2019}. Note that the equivariant ansatz has $2p$ parameters in total compared to the $7p$ of the non-symmetric ansatz.

\begin{figure}
\centering
\includegraphics{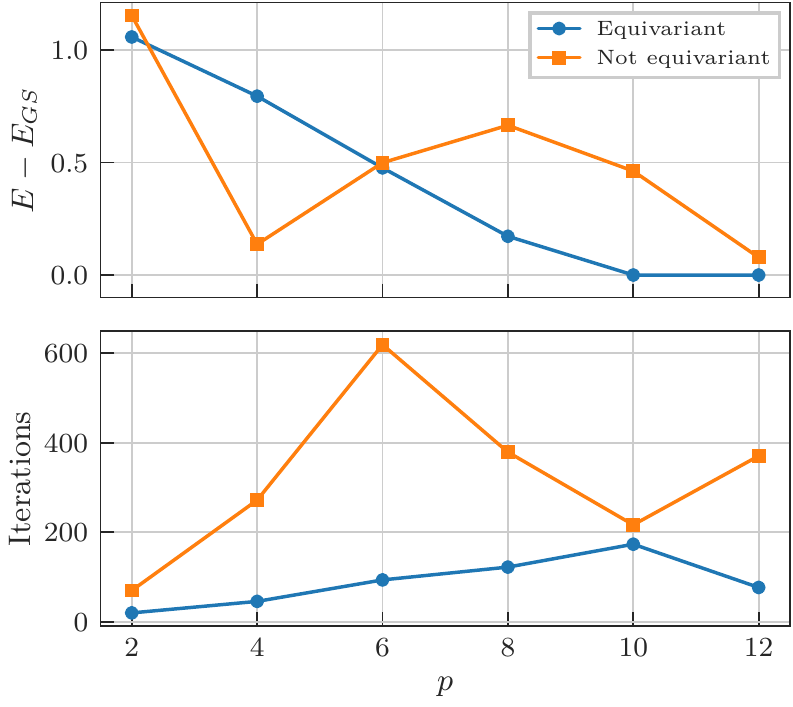}
\caption{
Comparison of the equivariant and the non-equivariant ansätze for the Heisenberg model on $N=10$ spins. For every number of layers $p$, we perform 20 experiments with random initializations of the parameters and plot the mean achieved energy (top) and the number of iterations necessary to reach it (bottom). We see that, for large enough $p$, the equivariant ansatz performs better in both figures of merit. For $p$ small, however, the approximation to the ground state is worse.}
\label{fig:Heis_example}
\end{figure}

We compare the performance of the two ansätze using the Heisenberg model with $N=10$ spins and different values for the number of layers $p$. As in the TFIM numerics, we optimize using the L-BFGS optimizer until convergence is reached and average the outcomes over 20 random initializations of the circuit parameters to arrive at the data shown in Fig.~\ref{fig:Heis_example}.

We see that the equivariant ansatz reaches its minimum energy faster across all depths. However, only for a large enough depth $p$ it can outperform the non-equivariant ansatz in terms of the energy expectation value. This nicely shows the trade-off between expressivity and symmetrization, as for small $p$ the increased expressivity of the non-equivariant ansatz is at an advantage over the equivariance of the symmetrized ansatz. However, as soon as the equivariant ansatz becomes sufficiently expressive, the picture reverses.

\begin{figure}
\centering
\includegraphics{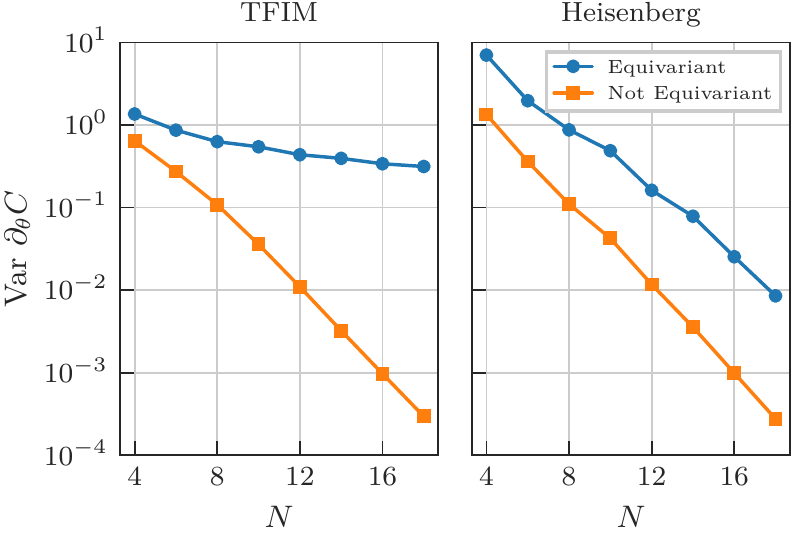}
\caption{Comparison of the barren plateaus phenomenon for the TFIM and the Heisenberg model considered in Figs.~\ref{fig:TFIM_example} and \ref{fig:Heis_example}. The derivative is computed with respect to the rotation
angle of the first qubit in the first layer using as observable $Z_1 Z_2$ with $p=80$ for the TFIM and $p=40$ for the Heisenberg model. The variance is calculated using 1000 randomly sampled parameters from a uniform distribution $[0,2\pi]$.}
\label{fig:BP_TFIM_Heis}
\end{figure}

\subsubsection*{Barren plateaus}
In a further numerical experiment that uses the same setups as above -- TFIM and Heisenberg model -- we analyze the influence of our symmetrization procedure on the barren plateaus phenomenon~\cite{Holmes_2022,McClean_2018}. Symmetrization reduces the expressivity of the ansatz by reducing the number of free parameters and additionally alters the dynamical Lie algebra associated with the ansatz generators, which is known to be intimately related to barren plateaus~\cite{BP_OptControl}. 
We therefore expect that the equivariant ansatz will have larger gradients in the applications studied here. 
In Fig.~\ref{fig:BP_TFIM_Heis}, we can indeed observe that the barren plateaus phenomenon is mitigated by the symmetrization procedure. 
For the transverse-field Ising model, the gradient decay of the equivariant ansatz is consistent with the polynomial scaling predicted by the reduced dimension of the dynamical Lie algebra found in Ref.~\cite{BP_OptControl}, whereas the non-equivariant ansatz shows an exponential decay. In the case of the Heisenberg model, the equivariant ansatz shows an exponential decay like its non-equivariant counterpart, but the gradient magnitude is enhanced and we see signs of a slightly slower decay exponent.

\subsubsection*{Tic-tac-toe LTFIM}
The last model we analyze is a variant of the \emph{longitudinal-transverse-field Ising model (LTFIM)} defined on a two-dimensional lattice of nine sites.
This model is constructed such that it has the same geometric symmetry as the tic-tac-toe example that we encountered above. We can write the Hamiltonian as
\begin{align}
    {H}_{\text{TTT}}=H_{ZZ} +\sum_{i=1}^{9}X_i+\sum_{i=1}^{9}Z_i,
\end{align}
where $H_{ZZ}$ determines the interaction of the spins. As shown in Fig.~\ref{fig:TTT_VQE_graph}, we distinguish three families of edges through their interaction strength: those in the contour of the lattice (triple lines in the figure), those along the diagonals (single lines), and those at the inside (double lines).
\begin{figure}
    \centering
    \includegraphics[width=.35\textwidth]{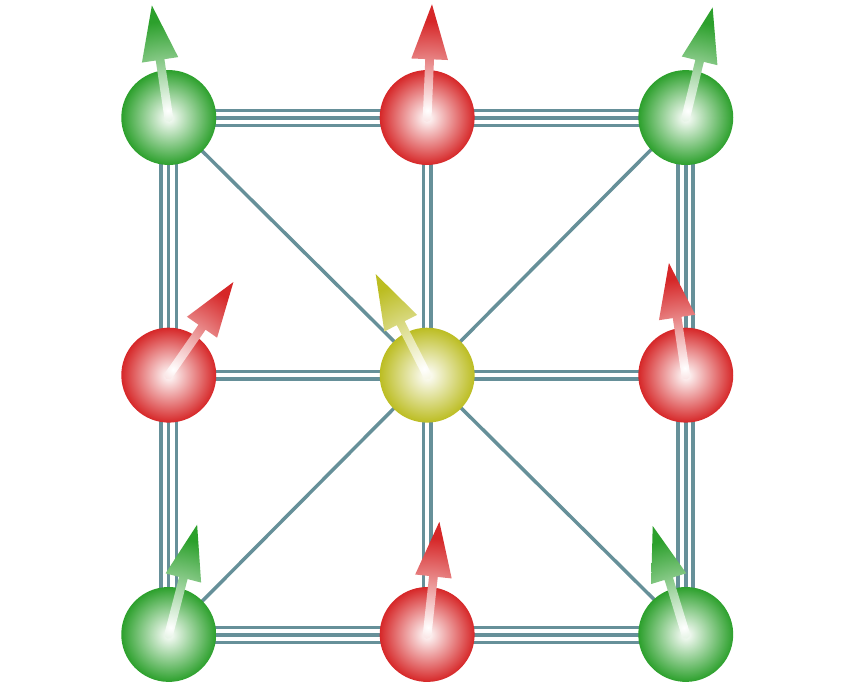}
    \caption{Graph for the LTFIM lattice model. The different colors used for the vertices emphasize the fact that sites that share the same color also share the same parameter in the one-qubit gates, and similarly for the different thicknesses of edges used for the two-qubit gates.}
    \label{fig:TTT_VQE_graph}
\end{figure}
So we choose $H_{ZZ}$ to be
\begin{align}
    H_{ZZ} = H_{\text{cont}} + \frac{1}{2} H_{\text{inside}} + \frac{3}{2} H_{\text{diag}},
\end{align}
where every term is a sum of the $ZZ$ interactions for the given set of edges 
\begin{align}
    H_{A}=\sum_{(i_1,i_2) \in A}Z_{i_1}Z_{i_2}.
\end{align}
This model has clear geometric symmetries which are the same we have used in the tic-tac-toe learning task. In Fig.~\ref{fig:TTT_VQE_graph}, we have represented the qubits that are equivalent under symmetry transformation with the same colors. Note that we have three independent families of qubits and also of edges. 

For this model, we use an ansatz which is composed of $p$ repeated layers of entangling $ZZ$ rotations for each edge in Fig.~\ref{fig:TTT_VQE_graph}, Pauli-$X$ and Pauli-$Z$ rotation on each qubit. The symmetrization procedure will enforce that the gate parameters of the gates are the same for equivalence classes of edges for the entangling gates and for equivalence classes of spins for the single-qubit gates.
Thus in the non-symmetric case, we have $(16+9+9)p$ free parameters in total, while in the symmetric one we have a reduction to $(3+3+3)p$ free parameters. The invariant state vector $ {\ket{+}^{\otimes N}}$ is used as the initial state.

We numerically tested the performance of the two ansätze in a manner similar to that of previous models as shown in Fig.~\ref{fig:TTTLTFIM_plot}. We can conclude that in this model, the equivariant ansatz performs better on average than the non-equivariant one, both in achieved energy and number of needed iterations for the circuit depths we tested. However, we see that there is an uptick in the energy achieved by the equivariant ansatz for larger layers which probably indicates that this advantage is not stable against increasing the depth of the circuit. This underscores the fact that symmetrization is usually expected to be helpful, but it still is a tool that needs to be treated with care as it is no panacea.

\begin{figure}
\centering
\includegraphics{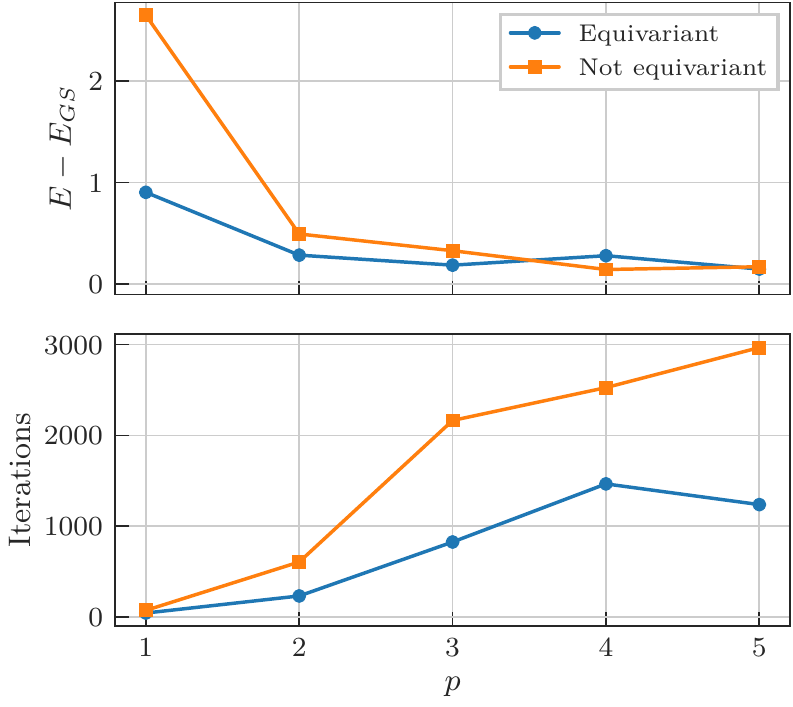}
\caption{
Comparison of the equivariant and the non-equivariant ansätze for the TTT-LTFIM model. For every number of layers $p$, we perform ten experiments with random initializations of the parameters and plot the mean achieved energy (left) and the number of iterations necessary to reach it (right). For smaller depth of the circuit, the equivariant ansatz performs better than the non-equivariant one. This advantage seems to not persist to higher depths.}
\label{fig:TTTLTFIM_plot}
\end{figure}

\subsubsection*{Discussion}

In the examples shown above, we examined different figures of merit, such as the final energy, the required iterations for the optimizer, and barren plateaus, to see if the use of an equivariant ansatz can yield better performance. 
We verified that in some cases it actually does, especially in the number of iterations necessary to reach a solution. We attribute this to the reduction in the number of free parameters of the model, which simplifies the underlying optimization problem. However, the picture is less clear when looking at the achievable minimal energy. In our numerics, we observed the trade-off between the gain in equivariance of the ansatz on one hand and the reduction in expressivity on the other hand. Especially at small circuit depths, it can be advantageous to have a non-equivariant ansatz that can explore larger portions of the Hilbert space, whereas at larger circuit depths the restriction of the expressivity to the relevant subspace of the Hilbert space can help.

This further motivates the use of carefully engineered symmetry-breaking in conjunction with equivariant ansätze, especially at small circuit depths as was explored in Ref.~\cite{ParkBrokenSymm}. A theoretical underpinning to this reasoning was given in Ref.~\cite{Bravyi_2020}, where it was shown that an ansatz that preserves parity symmetry, such as QAOA with a shallow circuit, may be an obstacle for the preparation of the ground state of certain Hamiltonians. It is therefore of essence to further clarify the conditions under which an equivariant ansatz can lead to better performance in ground state preparation problems.

\section{Summary and outlook} 

In this work, we have laid the foundations for the construction of variational quantum learning models that make predictions invariant under a symmetry transformation of the data. 
We have shown that the embedding of the data into the Hilbert space of the quantum system plays a crucial role and that it has to be chosen suitably to induce a meaningful unitary representation of the data symmetry on the level of the Hilbert space. We 
have provided embeddings that allow doing this for the most important symmetries encountered in contemporary learning scenarios, namely permutation-type symmetry, and also 
embeddings that induce a meaningful representation of the Lie group $\O(3)$ of orthogonal spatial transformations. 

With the unitary representation of the symmetry at hand, we show how elementary results from representation theory can be used to construct equivariant gatesets from standard gatesets used for the construction of variational quantum learning models. Armed with these gatesets, the construction of invariant re-uploading models becomes possible: alternating layers of equivariant data embeddings and equivariant trainable blocks applied to a symmetry-invariant initial state yield invariant predictions when evaluated on a symmetry-invariant observable. In this way, we give both a blueprint and tools for the construction of invariant variational quantum machine learning models. Moreover, using equivariant gatesets is a much needed building block that allows us to inform decisions about how to construct quantum learning models which is a first step to a solution of the quantum model selection problem. To increase the applicability of these tools, we also outlined the pitfalls that one should avoid when using them.

The numerical experiments we have conducted on the tic-tac-toe toy example and the autonomous vehicles toy problem have confirmed our expectation that invariant learning models indeed have better generalization capabilities, as their expressivity is constrained to a set of output functions that include some knowledge about the underlying data. We have also ensured that we did not cherry-pick the results by comparing random invariant model architectures with their non-invariant counterparts where we have observed the same results.

As the existence of a unitary symmetry representation is sufficient for the construction of equivariant gatesets, these can also be applied to problems outside the realm of variational quantum machine learning. We paradigmatically explore this possibility by comparing equivariant with non-equivariant ansatz circuits for ground state type problems where conserved quantities of the Hamiltonian give rise to a symmetry. Our analysis on the transverse-field Ising model, the Heisenberg model and a variant of the longitudinal-transverse-field Ising model with geometric symmetry allow us to conclude that equivariant ansätze can be helpful in this application as well. They often allow to reach a better energy estimate in fewer iterations and help to alleviate the problem of barren plateaus. Nevertheless, the picture is not as clear as in the learning scenario which we discussed in detail.

We have studied the question of equivariant quantum embeddings for the case of $\O(3)$. But there is no fundamental reason why other Lie groups that constitute a data symmetry should not also be amenable to equivariant quantum embeddings. While this is a somewhat exotic direction from the perspective of real-world learning tasks, we expect that future research into the interplay between data embeddings in the quantum system's Lie algebra and unitary representations of the symmetry on the level of the Hilbert space will allow us to learn much more about the inner workings of variational quantum learning models.

Another interesting direction for future research is to prove rigorous results about the generalization capabilities of invariant quantum machine learning models, with the aim to find quantitative expressions of how much exploiting a particular symmetry helps to better solve the learning task at hand.
One should also note that we have only run numerical experiments on maximally nine qubits for the learning problems, as our computational resources have been limited. In the future, it will be interesting to compare invariant variational quantum learning models to classical models on more realistic learning tasks.

Finally, it is our hope that this work stimulates further research efforts aimed at exploiting symmetry in variational quantum algorithms in the context of quantum machine learning and beyond.

\section*{Acknowledgments}
The authors would like to thank Hakop Pashayan and Regina Kirschner for insightful discussions and Matthias Caro, Simon Marshall, Sergi Masot, Franz Schreiber, Adri\'an P\'erez-Salinas, and Andrea Skolik for useful comments on an earlier version of this manuscript.
F.~A.\ acknowledges the support of the Alexander von Humboldt Foundation. We thank the BMBF (Hybrid), the BMWK (PlanQK), the QuantERA (HQCC), the Munich Quantum Valley (K8) and the Einstein Foundation (Einstein Research Unit on Quantum Devices) for their support. We also thank the Porsche Digital GmbH. 

\section*{Author contributions}
J.~J.~M.\ conceived and supervised the project. J.~J.~M.\, E.~G.~F.\ and F.~A.\ developed the theory for equivariant embeddings and gate symmetrization. M.~M.\ and F.~A.\ conducted the quantum machine learning experiments based on ideas by F.~A., M.~M., A.~W. A.~A.~M.\ conceived and conducted the ground state search experiments. A.~W.\ and J.~E.\ supported research and development. All authors contributed to the writing of the manuscript under the lead of J.~J.~M.

\section*{Data availability}
Code for implementations and data of the numerical experiments conducted in this work will be made available upon reasonable request.

\section*{Note added}
During the finalization of this manuscript, Ref.~\cite{larocca2022group-invariant} was published which has conceptual overlap with our work, especially in the definition of invariant and equivariant models. However, the authors of Ref.~\cite{larocca2022group-invariant} consider a different setting, where learning is performed on quantum states rather than classical data, and focus on binary classification. Our work discusses how and when symmetries on the level of the quantum system arise when dealing with classical data symmetries and how we can then exploit these to construct invariant learning models. As our techniques can also be readily applied to the setting of Ref.~\cite{larocca2022group-invariant}, we give a way of constructing equivariant architectures as left as a future direction in their work.

\bibliography{main}

\onecolumngrid
\clearpage
\appendix
\section{Additional figures}
\FloatBarrier
\begin{figure}
    \centering
    \includegraphics{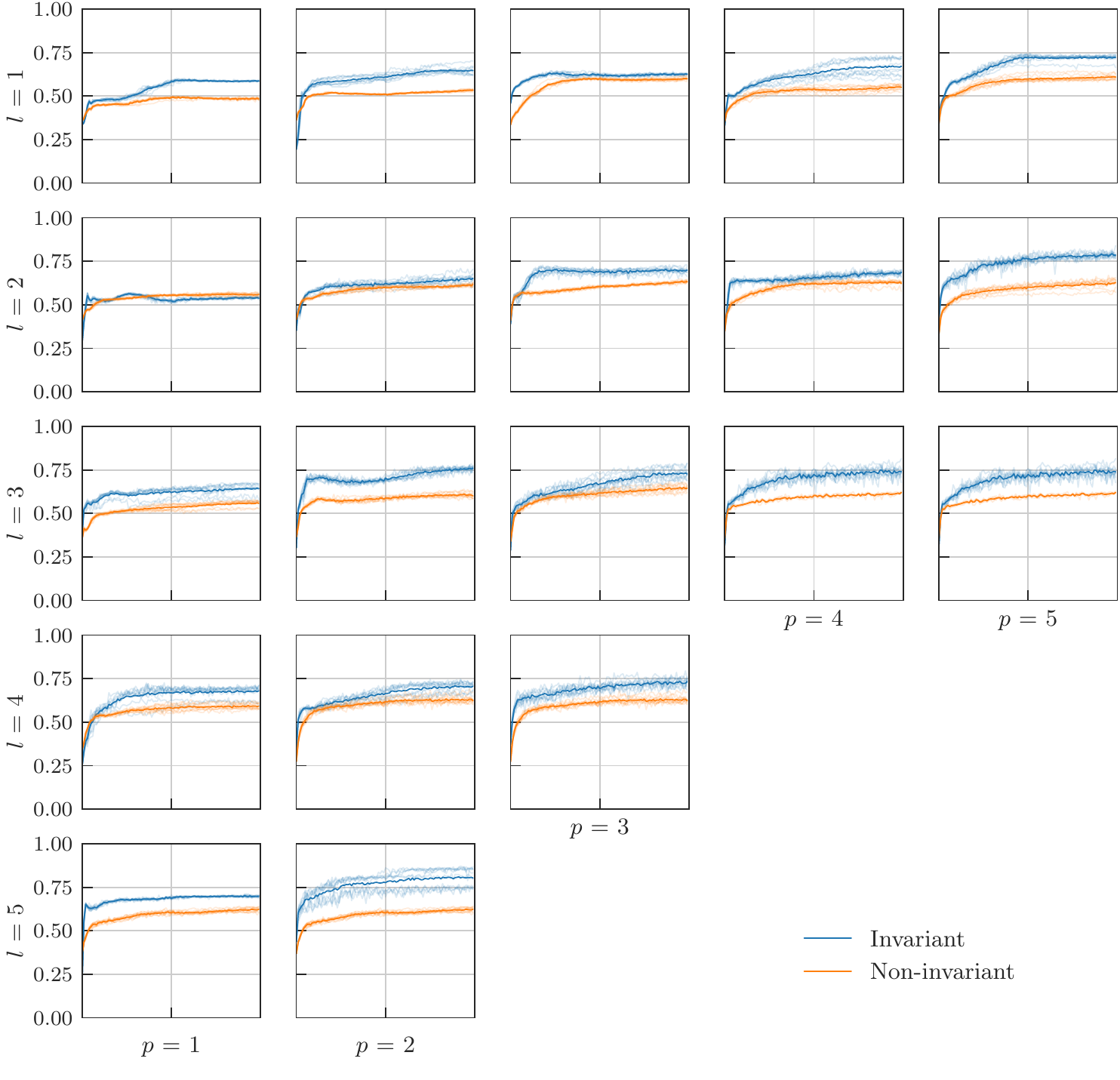}
    \caption{Plots of the sweep over different values of $l$ and $p$ for the classification of tic-tac-toe games arranged in a grid. Shown on the y-axis are the mean test accuracies for both the invariant and the non-invariant models, where the average over the ten runs is plotted in a darker color. Models are composed of $l$ layers, each of which contains one data encoding followed by $p$ times the standard layout \enquote{cemoid}. Optimizations were performed for $100$ epochs with $30$ steps each. At every step, the gradient was computed on $15$ games with a total training set size of $15 \times 30 = 450$. When an epoch ended, the training set was shuffled and the accuracy was evaluated on a test set of $600$ random games, which was fixed throughout a simulation.
    We can see that the invariant model clearly outperforms the non-invariant model. Not simulated were circuits for the values $(4, 4), (4, 5), (5, 3), (5, 4)$, and $(5,5)$ due to their growing computational demand. Further details can be found in the main text.}
    \label{fig:lp_sweep_full_ttt}
\end{figure}

\begin{figure}
    \centering
    \includegraphics{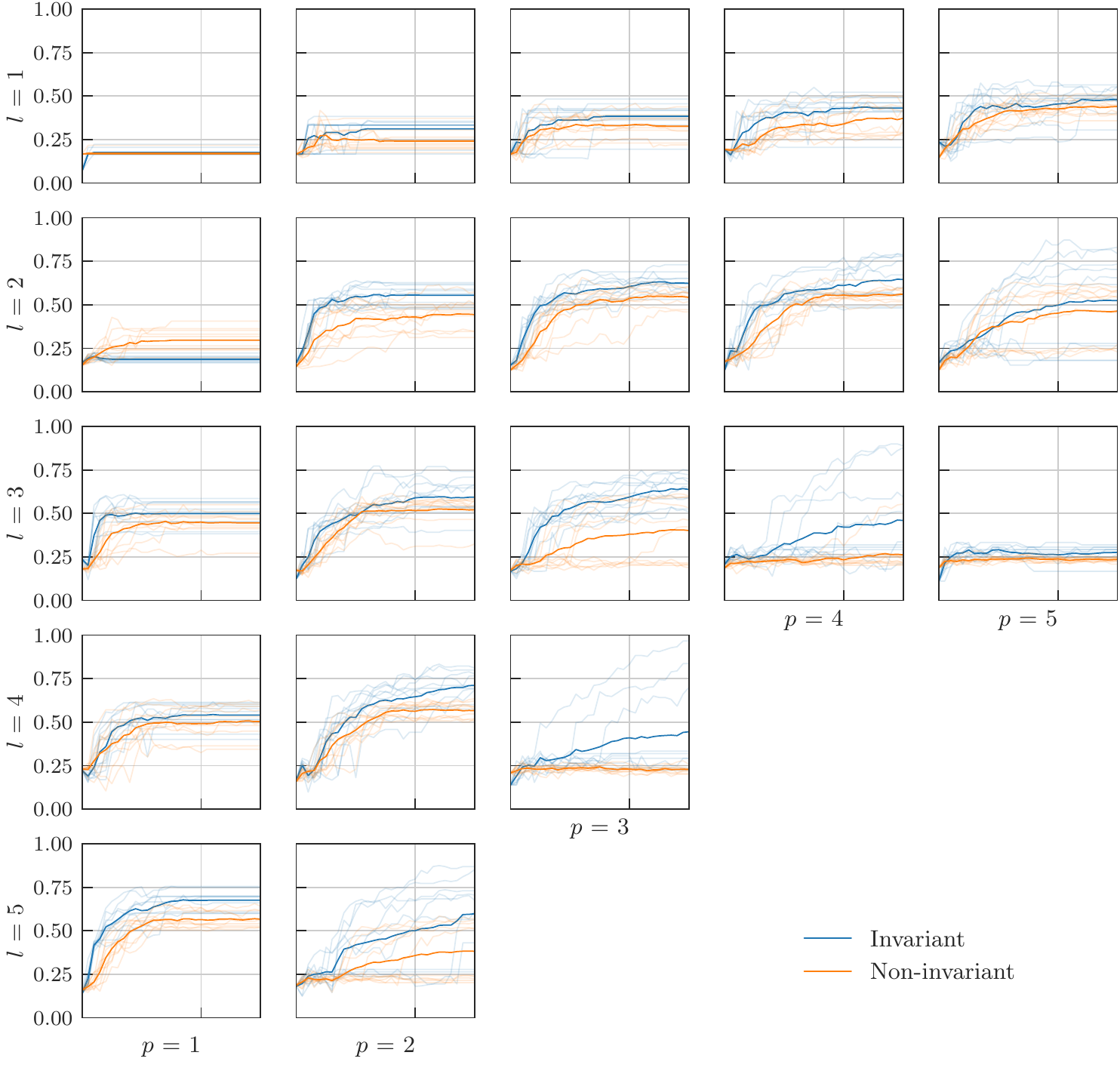}
    \caption{Plots of the sweep over different values of $l$ and $p$ for the classification of autonomous vehicle scenarios arranged in a grid. Shown on the y-axis are the mean test accuracies for both the invariant and the non-invariant models, where the average over the ten runs is plotted in a darker color. Models are composed of $l$ layers, each of which contains one data encoding followed by $p$ times the standard layout \enquote{cemoid}. Optimizations were performed for $30$ steps with no epochs on the same training set of $60$ situations that were chosen randomly at the start of each simulation. The accuracies were evaluated after each step on a test set of $130$ random but unique situations, which was fixed throughout a simulation.
    We can see that the invariant model clearly outperforms the non-invariant model. Additionally, it can be explicitly observed that for high $(l, p)$ values, the expressivity is empirically too high as accuracies drop sharply for both models. Not simulated were circuits for the values $(4, 4), (4, 5), (5, 3), (5, 4)$, and $(5,5)$ due to their growing computational complexity. Further details can be found in the main text.}
    \label{fig:lp_sweep_full_porsche}
\end{figure}

\end{document}